\documentclass[11pt]{article}

\usepackage{amsfonts, amsthm}
\usepackage{amssymb}
\usepackage{amsmath}  
\usepackage{graphicx, comment, xcolor}
\usepackage{authblk}
\usepackage{mathrsfs}
\usepackage[colorlinks=true, allcolors=blue]{hyperref}
\usepackage{cleveref}
\usepackage{mathtools}

\usepackage[T1]{fontenc}
\allowdisplaybreaks

\usepackage[shortlabels]{enumitem}
\usepackage{tikz}
\usetikzlibrary{positioning,calc}
\usetikzlibrary{intersections}
\tikzset{every picture/.style={line width=0.75pt}} 

\usepackage{fullpage}

\newtheorem{theorem}{Theorem}
\newtheorem{lemma}[theorem]{Lemma}
\newtheorem{corollary}[theorem]{Corollary}
\newtheorem{proposition}[theorem]{Proposition}

\newcommand{\tr}{\mathrm{tr}}
\newcommand{\ve}{\mathbf{v}}
\newcommand{\w}{\mathbf{w}}

\newcommand{\ac}{\mathbf{a}}

\newcommand{\R}{\mathbb{R}}
\newcommand{\C}{\mathbb{C}}

\newcommand{\E}{\mathcal{E}}

\newcommand{\Lc}{\mathcal{L}}
\newcommand{\cc}{\mathbf{c}}
\newcommand{\KMB}{\text{\,KMB}}

\newcommand{\cW}{\mathcal{W}}
\newcommand{\cF}{\mathcal{F}}

\newcommand{\cH}{\mathcal{H}}
\newcommand{\dd}{\text{\rm{d}}}
\newcommand{\PiInner}[2]{\pi_{#1}^{#2}}

\newcommand{\ket}[1]{|#1\rangle}
\newcommand{\bra}[1]{\langle#1|}
\newcommand{\ketbra}[2]{|#1\rangle\langle#2|}

\title{Monotonicity of the von Neumann Entropy under\\ Quantum Convolution}

\author[1]{Salman Beigi\thanks{salman.beigi@gmail.com}}
\author[2]{Hami Mehrabi\thanks{hamimhrb@gmail.com}}
\affil[1]{\it \footnotesize School of Mathematics, Institute for Research in Fundamental Sciences (IPM), P.O. Box 19395-5746, Tehran, Iran}
\affil[2]{\it \footnotesize  School of Electrical and Computer Engineering, Cornell University, Ithaca, NY, 14850, USA}

\date{}

\begin{document}
\maketitle

\begin{abstract}
The quantum entropy power inequality, proven by K\"onig and Smith (2012), states that $\exp(S(\rho \boxplus \sigma)/m)\geq  \frac 12 (\exp(S(\rho)/m) + \exp(S(\sigma)/m))$ for two $m$-mode bosonic quantum states $\rho$ and $\sigma$. One direct consequence of this inequality is that the sequence $\big\{ S(\rho^{\boxplus n}): n\geq 1  \big\}$ of von Neumann entropies of symmetric convolutions of $\rho$ has a monotonically increasing subsequence, namely, $S(\rho^{\boxplus 2^{k+1}})\geq S(\rho^{\boxplus 2^{k}})$. In the classical case, it has been shown that the whole sequence of entropies of the normalized sums of i.i.d.~random variables is monotonically increasing. Also, it is conjectured by Guha (2008) that the same holds in the quantum setting, and we have $S(\rho^{\boxplus n}) \geq S(\rho^{\boxplus (n-1)})$ for any $n$. In this paper, we resolve this conjecture by establishing this monotonicity. We in fact prove generalizations of the quantum entropy power inequality, enabling us to compare the von Neumann entropy of the $n$-fold symmetric convolution of $n$ arbitrary states $\rho_1, \dots, \rho_{n}$ with the von Neumann entropy of the symmetric convolution of subsets of these quantum states. Additionally, we propose a quantum-classical version of this entropy power inequality, which helps us better understand the behavior of the von Neumann entropy under the convolution action between a quantum state and a classical random variable.
\end{abstract}

{\footnotesize
}

\section{Introduction}\label{Sec:Intro}

The Shannon--Stam inequality, also known as the Entropy Power Inequality (EPI), states that for two real-valued independent random variables $X$ and $Y$, which have probability density functions and finite (differential) entropies, i.e., $H(X),  H(Y)<+\infty$, we have
\begin{equation}\label{eq:classical_Shannon}
e^{2 H(X+Y)} \geq e^{2 H(X)} + e^{2 H(Y)}.
\end{equation}
This inequality was first introduced by Shannon in his seminal work in 1948~\cite{shannon1948mathematical}, although he did not provide a complete proof for it. Later, Stam proposed a complete proof for this inequality in~\cite{stam1959some}.

An immediate consequence of~\eqref{eq:classical_Shannon} is that for independent random variables $X$ and $Y$, we have 
$$H\Big(\frac{X+Y}{\sqrt{2}}\Big) \geq \frac{1}{2} \big(H(X) + H(Y)\big).$$
Furthermore, assuming that $X_1, X_2, \dots$ is a sequence of independent and identically distributed (i.i.d.)~random variables with finite entropies, and letting $S_n \coloneqq \frac{X_1 + \cdots + X_n}{\sqrt{n}}$ be their normalized sum, we obtain $H\big(S_{2^k}\big) \geq H\big(S_{2^{k-1}}\big)$. This result is consistent with the Central Limit Theorem (CLT), specifically its entropic version~\cite{Barron1986}, as we know that if $X_1$ is centered, then $S_n$ converges to a centered Gaussian random variable $Z$ with the same variance as $X_1$. The point is that this Gaussian random variable attains the maximum entropy among all random variables with the same variance. Thus, the subsequence $H\big(S_{2^k}\big)$ of the entropies converges \emph{monotonically} to $H(Z)$. 

Extending this monotonicity, Lieb conjectured in 1978 that $H(S_n) \geq H(S_{n-1})$ holds for any integer~$n$~\cite{lieb1978proof}. This conjecture was first resolved by Artstein, Ball, Barthe, and Naor in 2004~\cite{ABBN1} (See also~\cite{tulino2006monotonic} and~\cite{Courtade} for alternative proofs.)
Later, Madiman and Barron proposed a more general form of this monotonicity~\cite{madiman2007generalized}.
Specifically, they showed that for independent random variables $X_1, \dots, X_n$ with finite entropy (but not necessarily identical), and an arbitrary class $\mathcal{C}$ of subsets of the set $[n] \coloneqq \{ 1, \dots , n \}$, it holds that
\begin{equation}\label{eq:generalized_classical_Shannon}
\exp \Big( 2 H\big( X_1 + \cdots + X_n\big) \Big) \geq \frac 1r \sum_{\ve \in \mathcal{C}} \exp \bigg( 2 H\Big( \sum_{k \in \ve} X_k\Big) \bigg),
\end{equation}
where $r$ is the maximum number of times that an element $k\in [n]$ appears in subsets in $\mathcal C$.
Letting $X_1, \dots, X_n$ to be identical, and $\mathcal C$ be the class of all subsets of $[n]$ of size $n-1$ and $r=n-1$, the aforementioned monotonicity result follows after an appropriate scaling.

\bigskip

A quantum version of EPI was first proven by K\"onig and Smith~\cite{KS14}. They showed that for two $m$-mode bosonic quantum states $\rho$ and $\sigma$ with finite second moments, we have
\begin{equation}\label{eq:FirstQEPI-1}
S(\rho \boxplus_\eta \sigma ) \geq \eta S(\rho) + (1 - \eta) S(\sigma), 
\end{equation}
for all $\eta\in [0,1]$
and 
\begin{align}\label{eq:FirstQEPI-2}
    \exp\bigg(\frac{S(\rho \boxplus_\eta \sigma )}{m}\bigg) \geq   \eta \exp\Big(\frac{S(\rho)}{m}\Big) + (1- \eta)  \exp\Big(\frac{S(\sigma)}{m}\Big),
\end{align}
for $\eta=1/2$. 
Here, as will be discussed later, the quantum convolution $\rho \boxplus_\eta \sigma$ is defined in terms of the interaction of the two states $\rho, \sigma$ via a beam splitter with the transmissivity parameter $\eta$, and is really a quantum generalization of the convolution of density functions. 
Inequality~\eqref{eq:FirstQEPI-2} for arbitrary values of $\eta\in [0,1]$ was later proven in~\cite{de2014generalization}. See also~\cite{de2019entropy, de2018conditional, de2015multimode, PaTre2018} for other forms of the quantum EPI for bosonic systems, ~\cite{audenaert2016entropy} for a quantum EPI in the finite-dimensional case, and~\cite{Huber2017geometric} for a quantum EPI with classical registers.

Similarly to the classical setting, we can view the quantum CLT through the lens of the quantum EPI. The quantum CLT states that for a centered $m$-mode bosonic quantum state $\rho$ with finite second moments, the $n$-fold (symmetric) convolution of $\rho$ with itself denoted by $\rho^{\boxplus n}$, converges to a Gaussian quantum state $\rho_G$ with the same first and second moments as $\rho$~\cite{CH71}. This, in particular, implies that $\lim_{n \rightarrow \infty} S(\rho^{\boxplus n}) = S(\rho_G)$. Now, similarly to the classical case, one may ask whether this convergence is monotonic or not, i.e., whether we have
\begin{align}\label{eq:Guha}
    S\big(\rho^{\boxplus n}\big)\geq S\big(\rho^{\boxplus (n-1)}\big).
\end{align}
This inequality was first conjectured by Guha in 2008~\cite{guha2008multiple}. We note that as a consequence of the quantum EPI we have $S\big(\rho^{\boxplus 2^k}\big) \geq S\big(\rho^{\boxplus 2^{k-1}}\big)$, yet the quantum EPI does not imply~\eqref{eq:Guha}.

\subsection{Main results}

Our main result in this work is a generalization of the quantum EPI along the work of Madiman and Barron~\cite{madiman2007generalized}.

\begin{theorem}\label{Th:Quantum_Shannon}
Let $\rho_1, \dots, \rho_n$ be $m$-mode bosonic quantum states with finite second moments. Let $\mathcal{C}$ be an arbitrary collection of subsets of $ [n]$, and let $r$ be the maximum number of times an index in $[n]$ appears in subsets in $\mathcal{C}$, i.e., $r= \max_{k\in [n]} |\{\ve\in \mathcal C:\, k\in \ve\}|$. 
Then, we have
\begin{equation}\label{eq:generalizedQEPI}
\exp \bigg( \frac 1m  S(\rho^{\boxplus [n]}) \bigg) \geq \frac{1}{rn} \sum_{\ve \in \mathcal{C}} |\ve| \exp \bigg( \frac 1m  S(\rho^{\boxplus \ve}) \bigg),
\end{equation}
where for any $\ve = \{k_1, k_2, \dots, k_{|\ve|}\} \subseteq [n]$, we use the notation $\rho^{\boxplus \ve}$ to indicate the $|\ve|$-fold symmetric convolution of $ \rho_{k_1}, \dots, \rho_{k_{|\ve|}}$. 
In particular, for $\mathcal C=\mathcal{C}_{n-1}$ being the class of all subsets of $[n]$ of size $n-1$, we have
\begin{equation}\label{eq:Strong_Monotonicity}
\exp \bigg( \frac 1m  S\big(\rho^{\boxplus [n]}\big) \bigg) \geq \frac 1n \sum_{\ve \in \mathcal{C}_{n-1}}  \exp \bigg( \frac 1m  S\big(\rho^{\boxplus \ve}\big) \bigg).
\end{equation}
\end{theorem}

We note that~\eqref{eq:Strong_Monotonicity} in the case where all states are equal, meaning that $\rho_1=\cdots=\rho_n=\rho$, resolves the monotonicity conjecture~\eqref{eq:Guha}.

Our proof of this theorem is based on an inequality in terms of a quantum Fisher information, called the Kubo--Mori--Bogoliubov (KMB) Fisher information. We will later give the precise definition of the KMB~Fisher information, but briefly speaking, this Fisher information denoted by $I_{\KMB}(\rho)$, can be understood as the derivative of the von Neumann entropy, when perturbing the state by the quantum heat semigroup. This is why the proof of Theorem~\ref{Th:Quantum_Shannon} is based on our second main result, that is an inequality in terms of the KMB~Fisher information.

\begin{theorem}\label{Th:Shannon_Stam_KMB-intro}
Let $\rho_1, \dots, \rho_n$ be $m$-mode bosonic quantum states with finite second moments. Let $\mathcal{C}$ be an arbitrary collection of subsets of $[n]$ and let $r$ be the maximum number of times an index in $[n]$ appears in subsets in $\mathcal{C}$. Then, for any probability distribution $\mu$ on $\mathcal{C}$, we have
\begin{align}\label{eq:Shannon_Stam_KMB-intro}
	I_\KMB\big(\rho^{\boxplus [n]} \big) \leq rn\, \sum_{\ve \in \mathcal{C}} \frac{1}{| \ve |}  \mu_{\ve}^2  \,I_\KMB\big(\rho^{\boxplus \ve} \big).
\end{align}
\end{theorem}

We indeed prove a more general form of the above theorems in which classical random variables are also present besides the quantum states. We refer to Section~\ref{Sec:proof_mainres} for the statements of these generalized results.

\subsection{Proof techniques}

As mentioned above, our first step to prove Theorem~\ref{Th:Quantum_Shannon} is to reduce it into a Shannon--Stam type inequality at the level of the KMB~Fisher information as stated in Theorem~\ref{Th:Shannon_Stam_KMB-intro}. This reduction is based on the quantum de Bruijn identity, which expresses the von Neumann entropy $S(\rho)$ in terms of the integration of the KMB~Fisher information when the state $\rho$ evolves under the action of the quantum heat semigroup. 
The idea of using the quantum de Bruijn identity was first applied in~\cite{KS14} and is a standard tool in the proof of quantum EPIs.

We already face a challenge in the reduction of Theorem~\ref{Th:Quantum_Shannon} to Theorem~\ref{Th:Shannon_Stam_KMB-intro} via the quantum de Bruijn identity. The point is that, as we are working with arbitrary subsets, it is not clear how individual states should be evolved under the heat semigroup in a consistent way. To address this issue, we prove a generalization of Theorem~\ref{Th:Shannon_Stam_KMB-intro} which also involves classical registers. The point is that the action of the quantum heat semigroup can be understood as the convolution with a Gaussian random variable. Thus, bringing classical registers into the picture would rectify the problem of independently evolving various quantum states via the quantum heat semigroup. With this idea, we indeed reduce the proof of Theorem~\ref{Th:Quantum_Shannon} to that of a generalization of Theorem~\ref{Th:Shannon_Stam_KMB-intro}.

To prove (the generalization of) Theorem~\ref{Th:Shannon_Stam_KMB-intro}, we draw ideas from~\cite{madiman2007generalized}. A main idea in~\cite{madiman2007generalized} is to decompose a tensor product Hilbert space into a certain direct sum of orthogonal subspaces. This allows to express the norm of a vector in the tensor product space in terms of the norm of its projections on the orthonormal subspaces. Now the point
is that the KMB~Fisher information can be expressed as the norm of an operator called the \emph{score operator}, living in some Hilbert space. Another key observation is that 
the score operator behaves nicely under quantum convolution, and the above-mentioned decomposition is such that the corresponding projections of the score operators have operational meanings. These two facts can be used to bound $I_{\KMB}(\rho^{\boxplus [n]})$ in terms of the KMB~Fisher information for various states appearing on the right hand side of~\eqref{eq:Shannon_Stam_KMB-intro}. Nevertheless, to formalize these high-level ideas in the quantum case, we need to overcome two main challenges.

The first challenge is to view the score operator associated with a state $\rho^{\boxplus \ve}$ as an operator in some tensor product Hilbert space. The point is that $\rho^{\boxplus \ve}$ and its corresponding score operator are $m$-mode operators and lack an inherent tensor product structure. To overcome this difficulty, we generalize the idea of \emph{symmetric lifting map} first introduced in~\cite{beigi2023towards}. This map allows the lifts of the score operators of states of the form $\rho^{\boxplus \ve}$, for subsets $\ve \subseteq [n]$, to simultaneously live in a single tensor product Hilbert space in a quite natural way. 

The second challenge is that our generalized lifting map is not an isometry and
does not preserve the norm of score operators. The point is that the inner product based on which the KMB~Fisher information is defined, is hard to work with, and because of non-commutativity, does not satisfy some desired linearity properties. Our idea to overcome this challenge is to write the KMB~Fisher information as an integral over some other quantum Fisher informations that do satisfy the linearity property. In this way, the proof of Theorem~\ref{Th:Shannon_Stam_KMB-intro} reduces to the proof of the same theorem for other quantum Fisher informations that are easier to work with. 

Resolving the above two main challenges, we then apply the decomposition idea of~\cite{madiman2007generalized} and finish the proof of Theorem~\ref{Th:Shannon_Stam_KMB-intro}.

\subsection{Structure of the paper}
The rest of this paper is structured as follows. In Section~\ref{sec:Pre}, we review some fundamental definitions regarding bosonic quantum systems. Additionally, in Subsection~\ref{subsec:Con}, we discuss the definition of convolution in the quantum setting, and in Subsection~\ref{subsec:bruijn} review the notion of the quantum heat semigroup and the quantum de Bruijn identity. In Section~\ref{Sec:Methods}, we develop the main tools needed to prove the generalization of Theorem~\ref{Th:Shannon_Stam_KMB-intro}. Specifically, in Subsection~\ref{Sec:KMBInner}, we review the notion of the KMB~inner product and introduce an integral representation for the KMB~Fisher information, enabling us to work with linear inner products instead of the KMB~inner product. Moreover, in Subsection~\ref{Sec:LiftingMap}, we introduce our generalized symmetric lifting map, which is an essential tool in our arguments. Additionally, in Subsection~\ref{Sec:Decomp}, generalizing the work of~\cite{madiman2007generalized}, we discuss the method of the decomposition of a tensor product Hilbert space. After developing all these tools, we state the proof of our main results in Section~\ref{Sec:proof_mainres}.


\section{Preliminaries}\label{sec:Pre}
In this part, we review some basic definitions related to bosonic quantum systems. For a more detailed review, we refer the reader to~\cite{Serafini}.

Starting with a generic separable Hilbert space $\mathcal{H}$, an operator $T$ acting on $\mathcal{H}$, is called a trace class operator if $\tr\big(|T|\big) < \infty$ where $|T|=\sqrt{T^\dagger T}$. For two operators $T$ and $R$, their Hilbert--Schmidt inner product is defined as $\langle T,R \rangle \coloneqq \tr (T^\dagger R)$. This inner product induces a norm, and an operator with a finite Hilbert--Schmidt norm is called a Hilbert--Schmidt operator.

The Hilbert space of an $m$-mode bosonic quantum system is isomorphic to the space of all square-integrable, complex valued functions on $\R^m$, which we denote by $\mathcal{H}_m = L^2(\R^m)$. We note that $L^2(\R^m)=L^2(\R)^{\otimes m}$, so $\mathcal H_m=\cH_1^{\otimes m}$ has a natural tensor product structure.
An $m$-mode bosonic quantum state is a positive semi-definite operator acting on $\mathcal{H}_m$ with a trace equal to $1$, called a density operator. 

Fock basis is a standard orthonormal basis for $\mathcal H_m=\cH_1^{\otimes m}$ consisting of vectors of the form $\ket{\ell_1, \dots, \ell_m} = \ket{\ell_1}\otimes \cdots \otimes \ket{\ell_m}$ where $\ell_1, \dots, \ell_m$ are non-negative integers. The vectors $\ket{\ell_1, \dots, \ell_m}\in \mathcal H_m$  are often called number states. The, \emph{annihilation operators} $\ac_1, \dots, \ac_m$ and their adjoint $\ac_1^\dagger, \dots, \ac_m^\dagger$ called the \emph{creation operators}, are defined by their action on number states as
$$\ac_j \ket{\ell_1, \dots, \ell_m} = \sqrt{\ell_j} \ket{\ell_1, \dots,\ell_{j-1}, \ell_j-1, \ell_{j+1}, \dots, \ell_m},$$ 
and  
$$\ac^\dagger_j \ket{\ell_1, \dots, \ell_m} = \sqrt{\ell_j+1} \ket{\ell_1, \dots,\ell_{j-1}, \ell_j+1, \ell_{j+1}, \dots, \ell_m}.$$
Observe that $\ac_j, \ac_j^\dagger$ act on the $j$-th factor in $\cH_m=\cH_1^{\otimes m}$.
These operators satisfy the canonical commutation relations
\begin{align*}
    [\ac_j, \ac_k^\dagger] = \delta_{jk} I,\qquad   [\ac_j, \ac_k] = 0,
\end{align*}
where $[A, B] = AB - BA$ denotes the commutator of two operators, and $\delta_{jk}$ is the Kronecker's delta function. Moreover, $I$ is the identity operator that acts on $\mathcal H_m$. The operator $\sum_{j=1}^m \ac_j^\dagger \ac_j$ is called the number operator and its action on number states is given by
$$\Big(\sum_{j=1}^m \ac_j^\dagger \ac_j\Big)\ket{\ell_1, \dots, \ell_m} = ( \ell_1+\cdots + \ell_m) \ket{\ell_1, \dots, \ell_m}.$$

For any $z=(z_1, \dots, z_m)^\top \in \C^m$, the \emph{displacement operator} acting on $\mathcal{H}_m=\cH_1^{\otimes m}$ is given by 
$$D_z := \bigotimes_{j=1}^m \exp \big( z_j \ac_j^\dagger - \bar{z}_j \ac_j \big).$$ 
These operators play the role of shifting a random variable in the quantum setting as we have $D_z^\dagger \ac_j D_z = \ac_j + z_j$, and $D_z^\dagger \ac_j^\dagger D_z = \ac_j^\dagger + \bar{z}_j$ for any $1 \leq j \leq m$. Moreover, it can be verified that, for any $z , w \in \mathbb{C}^m$, we have
\begin{equation}\label{eq:productRule}
D_z D_w = e^{\frac{1}{2} (z^{\top} \bar w - \bar z^{\top} w)} D_{z+w}, \qquad \text{ and }\qquad D_z^\dagger = D_{-z}.
\end{equation}
Also, as a direct consequence of~\eqref{eq:productRule}, we get 
\begin{equation}\label{eq:productRuleN}
D_{w}^\dagger D_z D_w = e^{z^\top \bar w - \bar z^\top w} D_z.
\end{equation}

For any trace class operator $T$, the \emph{quantum characteristic function} of $T$, is defined as 
$$\chi_T(z) := \tr \big( T D_z \big),$$ for any $z \in \C^m$. The quantum characteristic function is sufficient to fully recover the operator as
\begin{equation}\label{eq:inverse_Char}
T = \frac{1}{\pi^m} \int_{\mathbb{C}^m} \chi_T(z) D_{-z} \dd^{2m} z.
\end{equation}

We say that $\rho$ has finite second-order moments if $\sum_{j=1}^m \tr(\rho \ac_j^\dagger \ac_j)<+\infty$. We should be careful about the interpretation of $\tr(\rho \ac_j^\dagger \ac_j)$ since $\ac_j^\dagger \ac_j$ is an unbounded operator. In this paper, by $\tr(\rho \ac_j^\dagger \ac_j)$ we indeed mean $ \tr(\ac_j\rho \ac_j^\dagger)$. Therefore, $\rho$ has finite second-order moments if $\sum_{j=1}^m \tr(\rho \ac_j^\dagger \ac_j) = \sum_{j=1}^m \tr(\ac_j \rho \ac_j^\dagger)<+\infty$, or equivalently  
$$\sum_{\ell_1, \dots, \ell_m=0}^{+\infty} (\ell_1+\cdots + \ell_m) \bra{\ell_1, \dots, \ell_n} \rho \ket{\ell_1, \dots, \ell_m}<+\infty.$$
Similarly, for any possibly unbounded operator $X$, by $\tr(\rho X^\dagger X)$ we mean $\tr(X\rho X^\dagger)$. Therefore, for a positive semidefinite operator $A$ we write $\tr(\rho A) = \tr(A^{1/2}\rho A^{1/2})$. Also, for an arbitrary self-adjoint operator $A$, we first express it as $A= A_+-A_-$ where $A_+, A_-$ are positive semidefinite, and then write $\tr(\rho A)= \tr(\rho A_+) - \tr(\rho A_-)$, when both  $\tr(\rho A_+), \tr(\rho A_-)$ are finite.
For more details on the definition of  $\tr (\rho A)$ for any operator $A$, we refer to~\cite[Section 2]{beigi2023towards} and references therein.

The von Neumann entropy of an $m$-mode bosonic quantum state $\rho$ is given by 
$$S(\rho) = -\tr(\rho \log \rho).$$
It can be verified that if $\rho$ has finite second-order moments, then $S(\rho)<+\infty$.

\subsection{Convolution operation}\label{subsec:Con}

In the classical setting, for two independent random variables $X$ and $Y$, the random variable $Z = \sqrt{\lambda} X + \sqrt{1-\lambda} Y$ is the normalized sum of $X$ and $Y$ with respect to the parameter $0 \leq \lambda \leq 1$. In this paper, we use the notation $Z = X \boxplus_\lambda Y$ to denote this normalized sum and refer to it as the classical convolution of $X, Y$. This is because the resulting probability density function of $Z$ is equal to the convolution of the density functions of $\sqrt{\lambda} X$ and $\sqrt{1-\lambda} Y$, assuming that $X$ and $Y$ are independent.

The counterpart of this action in the quantum setting is called \emph{quantum convolution}. Let $\rho$ and $\sigma$ be two $m$-mode bosonic quantum states. We define the convolution of $\rho$ and $\sigma$ with parameter $\eta \in [0,1]$ as
\begin{equation}\label{def:Quantum_Conv}
\rho \boxplus_\eta \sigma := \tr_2 \Big( U_\eta \, (\rho \otimes \sigma) \, U_\eta^\dagger \Big),
\end{equation}
where $U_{\eta}$ is the Gaussian unitary of beam splitter with transmissivity parameter $0 \leq \eta \leq 1$, which can be written as
\[
U_{\eta} := \exp\bigg(\arccos(\sqrt{\eta}) \sum_{j=1}^{m}\big(\ac_{j,1}^{\dagger} \ac_{j,2} - \ac_{j,1} \ac_{j,2}^{\dagger}\big)\bigg).
\]
We note that $U_\eta$ acts on $\cH_m\otimes \cH_m$ corresponding to $2m$ annihilation operators which we denote by $\ac_{j,1}, \ac_{j,2}$ for $j=1, \dots, m$.
Therefore, $\ac_{j,1}$ and $\ac_{j,2}$ are the $j$-th annihilation operators acting on the first and second subsystems, respectively. Also, the partial trace in~\eqref{def:Quantum_Conv} is taken with respect to the second subsystem.
It can be verified that $U_\eta$ transforms the annihilation and creation operators as
\begin{align} 
	\begin{cases} \label{BeamSplitter-quadratures}
		U_{\eta} \ac_{j,1} U_{\eta}^{\dagger} = \sqrt{\eta} \ac_{j,1} - \sqrt{1 - \eta} \ac_{j,2},\\
		U_{\eta} \ac_{j,2} U_{\eta}^{\dagger} =  \sqrt{1 - \eta} \ac_{j,1} + \sqrt{\eta} \ac_{j,2}. 
	\end{cases}
\end{align}
Using these relations, it is easily shown that 
\begin{equation}\label{eq:convolution-char}
\chi_{\rho \boxplus_{\eta} \sigma}(z) = \chi_{\rho}\big(\sqrt{\eta} z\big) \, \chi_{\sigma}\big(\sqrt{1 - \eta} z\big).
\end{equation}

In addition to the fully quantum setting, we can convolve a quantum state and a classical random variable. Let $\rho$ be an $m$-mode bosonic quantum state, and $X$ be a $\C^{m}$-valued random vector with probability density function $p_X$. The \emph{quantum-classical convolution} of $\rho$ and $X$ with parameter $t \geq 0$ is defined as
\begin{equation}\label{def:Quantum_Class_Conv}
\rho \star_t X = \int p_X(x) D_{\sqrt{t} x} \rho D_{\sqrt{t} x}^\dagger \dd^{2m} x.
\end{equation}
That is, $\rho\star_t X$ is again an $m$-mode state resulting from randomly displacing $\rho$ with the displacement parameter chosen according to $X$ after scaling with $\sqrt t$. For $t=1$, we remove the parameter $t$ in the quantum-classical convolution and denote it simply by $\rho\star X$. To the best of our knowledge, the definitions of quantum convolution, and quantum-classical convolution first appeared in~\cite{werner1984quantum}.

The characteristic function of the quantum-classical convolution can be computed as

\begin{align*}
\chi_{\rho \star_t X}(z) &= \tr \Big( ( \rho \star_t X) D_z \Big) \\
& = \int  p_X(x)  \tr \Big( \rho D_{\sqrt{t} x}^\dagger D_z D_{\sqrt{t} x} \Big) \dd^{2m}x \\
&= \tr (\rho D_z )  \int  p_X(x)  e^{\sqrt{t}(z^\top \bar x - \bar z^\top x)} \dd^{2m}x,
\end{align*}
where in the last line we use~\eqref{eq:productRuleN}. Motivated by this equation, we define the symplectic characteristic function of the $\C^{m}$-valued random vector $X$ as
\begin{align}\label{eq:chi-classic-D}
    \chi_X(z) = \int p_X(x)  D_\cc(z,x) \dd^{2m}x,\qquad \quad D_\cc(z,x) = e^{(z^\top \bar x - \bar z^\top x)}.
\end{align}
We understand the function $D_\cc(z,x)$ as the classical counterpart of displacement operators. With this notation in hand, we have 
\begin{align}\label{eq:char-q-c-conv}
\chi_{\rho \star_t X}(z) = \chi_\rho(z) \, \chi_X\big(\sqrt t z\big).
\end{align}
Also, it is evident that for two $\C^{m}$-valued, independent random vectors $X$ and $Y$, it holds that
$$\chi_{X\boxplus_\lambda Y}(z) = \chi_{X}\big(\sqrt \lambda z\big) \, \chi_{Y}\big(\sqrt{1-\lambda} z\big).$$

\medskip

The classical and quantum convolutions behave nicely with respect to each other. For $m$-mode bosonic quantum states $\rho, \sigma$, independent $\C^{m}$-valued random vectors $X, Y$, $\eta \in [0,1]$ and $t_1, t_2 \geq 0$, we have 
\[
	(\rho \star_{t_1} X) \boxplus_\eta  (\sigma \star_{t_2} Y) = (\rho \boxplus_\eta \sigma) \star_s (X \boxplus_\lambda Y),\qquad \quad s = t_1  \eta + t_2  (1-\eta) \text{ and } \lambda = t_1  \eta /s.
\]

We define the symmetric convolution of $n$ quantum states inductively. Letting $\rho_1, \dots, \rho_n$ be $m$-mode bosonic quantum states, we define their symmetric convolution by
$$\rho_1 \boxplus \cdots \boxplus \rho_n \coloneqq \big( \rho_1 \boxplus \cdots \boxplus \rho_{n-1} \big) \boxplus_{1 - \frac 1n} \rho_n.$$ 
It can be easily verified that 
\begin{align}\label{eq:symm-conv-char}
    \chi_{\rho_1 \boxplus \cdots \boxplus \rho_n }(z) = \prod_{k=1}^n \chi_{\rho_k}(z/\sqrt n).
\end{align}    
For simplicity of notation we often denote $\rho_1 \boxplus \cdots \boxplus \rho_n$ by $\rho^{\boxplus [n]}$. 
Also, for a subset $\ve\subseteq [n]$ we let $\rho^{\boxplus \ve}$ be the symmetric convolution of states $\rho_k$ with $k\in \ve$. These notations are unambiguous since by~\eqref{eq:symm-conv-char} symmetric convolution is invariant under permutations of the states. 

Using similar notations, for $\C^{m}$-valued independent random vectors $X_1, \dots, X_{n}$, we define their symmetric convolution by
$$X^{\boxplus [n]}=X_1 \boxplus \cdots \boxplus X_{n} \coloneqq \big( X_1 \boxplus \cdots \boxplus X_{n-1} \big) \boxplus_{1 - \frac{1}{n}} X_n.$$ 
We note that $X^{\boxplus [n]}$ is the same as the symmetric normalized sum $\frac{X_1 + \cdots + X_n}{\sqrt n}$. For a subset $\ve\subseteq [n]$, the random variable $X^{\boxplus \ve}$ is defined similarly.

An interesting property of convolution is that it interacts smoothly with the commutator action involving annihilation and creation operators.

\begin{lemma}\cite[Lemma 1]{beigi2023towards} Let $\rho$ and $\sigma$ be two $m$-mode bosonic quantum states with finite second-order moments, and $X$ be a $\C^{m}$-valued random vector with finite second-order moments. Let $0 \leq \eta \leq 1$ and $t \geq 0$ be the convolution parameters. Then, we have
\begin{equation}\label{eq:Comm_Conv}
\sqrt{\eta} \big[\ac_{j}, \rho \boxplus_{\eta} \sigma\big] = \big[\ac_{j,1},\, \rho\big] \boxplus_{\eta} \sigma, \quad\quad   \big[\ac_{j}, \rho \star_t X \big] = \big[\ac_{j,1},\, \rho\big] \star_t X,
\end{equation}
where $\ac_j$ is the $j$-th annihilation operator, and $\ac_{j, 1}$ is the $j$-th annihilation operator of the first subsystem.
\label{lem:from-beigi2023towards} 
\end{lemma}

Note that, since we assume $\rho$ has finite second-order moment, $\big[\ac_{j,1},\, \rho\big]$ is a trace class operator. Moreover, it is not hard to generalize the definition of convolutions in the previous subsection for all trace class operators. This is why the right hand sides in~\eqref{eq:Comm_Conv} are well-defined.   

The proof of the first equation in~\eqref{eq:Comm_Conv} is given in~\cite[Lemma 1]{beigi2023towards}. The proof of the second equation is similar and is skipped.

\subsection{Quantum de Bruijn identity}\label{subsec:bruijn}

In the classical setting, de Bruijn identity expresses the entropy of a random variable in terms of its Fisher information. A variant of the de Bruijn identity states that for a random variable $X$ with finite variance, it holds that
\begin{equation}\label{eq:classical_deBruijn}
D(X \| Z) = \int_{0}^{1} \frac{J(\sqrt{t} X + \sqrt{1-t} Z)}{2t} \dd t,
\end{equation}
where $Z$ is an independent Gaussian random variable with the same first and second moments as $X$, and $D(X \| Z)$ denotes the Kullback--Leibler~divergence.
Moreover, $J(\cdot)$ is the Fisher information distance given by
\[
J(Y) = \text{Var}(Y) I(Y) - 1, 
\] 
where
\[I(Y) = \mathbb{E}_Y\bigg[\Big(\frac{\dd}{\dd y} \log p_Y(y)\Big)^2\bigg],\]
is the Fisher information.
We observe that the Fisher information $I(Y)$ is the norm of the score function $\frac{\dd}{\dd y} \log p(y)$ with respect to the inner product $\langle h , g \rangle_Y = \mathbb{E}_Y [\overline{h(y)} g(y)]$. In the classical setting,~\eqref{eq:classical_deBruijn} enables us to reduce the proof of entropic inequalities to that of inequalities in terms of Fisher information. 

Returning to the quantum setting, we have a similar identity that relates the von Neumann entropy to a quantum Fisher information. To state this equation, we first need to define the \emph{quantum heat semigroup}. To this end, define the Lindbladian $\Lc$ acting on $m$-mode bosonic quantum states by 
$$\Lc(\rho) =  \sum_{j=1}^m \big[ \ac_j^\dagger , [\ac_j, \rho ]\big].$$ Then, the heat semigroup $\{\Phi_t:\, t\geq 0\}$ is given by \(\Phi_t(\rho) \coloneqq e^{-t \Lc}(\rho)\) and consists of completely positive and trace-preserving (CPTP) maps.

The action of the quantum heat semigroup is easily understood by looking at characteristic functions. 
It is not hard to verify that~\cite{KS14}
$$\chi_{\Phi_t(\rho)}(z) = e^{-t |z|^2} \chi_\rho(z),$$
where for $z\in \mathbb C^m$ we use $|z|^2 = \sum_{j=1}^m |z_j|^2$. Using this equation, the action of the quantum heat semigroup can also be expressed in terms of quantum-classical convolution. Let $Z$ be a $\C^{m}$-valued centered Gaussian random vector with covariance matrix $2 I_{2m}$. Then, using~\eqref{eq:char-q-c-conv} we find that 
$$\Phi_t(\rho) = \rho \star_t Z.$$

With the quantum heat semigroup, we can derive the quantum version of de Bruijn identity, as we have
\begin{align}
\frac{\dd}{\dd t} S\big(\Phi_t(\rho)\big) \bigg|_{t = 0} &= -\frac{\dd}{\dd t} \tr \big( \Phi_t(\rho) \log \Phi_t(\rho) \big) \bigg|_{t = 0} \nonumber  \\
& = \tr \big( \Lc(\rho) \log \rho\big) \nonumber \\
& = \tr \big( \rho \Lc(\log \rho)\big) \label{eq:first_KMB},
\end{align}
where in the last line we use the fact that $\Lc$ is a self-adjoint superoperator with respect to the Hilbert--Schmidt inner product. The last term in~\eqref{eq:first_KMB} is called the \emph{Kubo--Mori--Bogoliubov (KMB) quantum Fisher information}. Overall, we derive 
\begin{equation}\label{eq:deBruijn}
\frac{\dd}{\dd t} S\big(\Phi_t(\rho)\big) \bigg|_{t = 0} =  I_\KMB (\rho), \quad\quad  I_\KMB (\rho) = \sum_{j=1}^m \tr \Big( \rho \big[ \ac_j^\dagger , [\ac_j, \log \rho ]\big] \Big).
\end{equation}
Equation~\eqref{eq:deBruijn} is called the \emph{quantum de Bruijn identity}. Analogously to the classical setting, this identity is used to reduce inequalities on von Neumann entropy to inequalities on KMB quantum Fisher information. The reduction is based on the following lemma characterizing the asymptotic behavior of entropy under the quantum heat semigroup.

\begin{lemma}\cite[Corollary III.4]{KS14} \label{lem:entropy-heat-semigroup}
Let $\rho$ be an $m$-mode quantum state with finite second-order moments. Also, let $Z$ be a Gaussian random variable with covariance matrix $2bI_{2m}$. Then, we have
$$\big|S(\rho\star_t Z) - m(c + \log b+ \log t)\big|=O\Big(\frac{1}{t}\Big),$$
where $c$ is a universal constant independent of $\rho$, $b$ and $t$. 
\end{lemma}

\section{Technical tools}~\label{Sec:Methods}

This section is dedicated to developing the tools required to prove a generalization of Theorem~\ref{Th:Shannon_Stam_KMB-intro}. In the first part, we introduce the KMB inner product and the KMB score operator, along with an integral representation of the KMB inner product. In the second part, we generalize the symmetric lifting map proposed in~\cite{beigi2023towards} to accommodate classical registers in addition to quantum registers. Finally, in the third part, we generalize the decomposition of tensor product Hilbert spaces introduced in~\cite{madiman2007generalized} to the quantum case.

\subsection{Fisher information and the score operator}\label{Sec:KMBInner}

In the classical setting, proofs of EPIs such as~\eqref{eq:generalized_classical_Shannon} are often based on the linear algebraic properties of the score function. Generalizing this idea, we would like to express the KMB~Fisher information $I_\KMB(\rho)$ as the squared norm of some score operator. To this end, we also need to introduce the KMB inner product.

For any $m$-mode bosonic quantum state $\rho$, and operators $T, R$ acting on $\mathcal{H}_m$, we introduce the KMB inner product\footnote{We note that $\langle \cdot, \cdot\rangle_{\rho,\KMB}$ is really an inner product if $\rho$ is faithful, yet if $\rho$ is not faithful we can consider it as an inner product on an appropriate subspace of operators.} as 
\begin{align} \label{eq:KMB-inner-nr}
\langle T,R\rangle_{\rho, \KMB} := \tr\big(T^{\dagger} \pi_\rho^{\psi_\KMB}(R)\big) = \tr\big(\pi_\rho^{\psi_\KMB}(T)^{\dagger} R\big),
\end{align}
where $\psi_\KMB(x,y) = \frac{x-y}{\log x - \log y}$ if $x\neq y$ and $\psi_{\text{KMB}}(x,x)=\frac{1}{x}$. Moreover, for any continuous function $f(x, y)$ the superoperator $\pi_\rho^{f}$ is defined as~\cite{bardet2017estimating} 
\begin{align}\label{eq:pi-rho-f} 
\pi_\rho^{f} := f(\mathcal{M}_{\ell, \rho}, \mathcal{M}_{r, \rho}),
\end{align}
with $\mathcal{M}_{\ell, \rho}, \mathcal{M}_{r, \rho}$ being the left and right multiplications by $\rho$, respectively, i.e., $\mathcal{M}_{\ell, \rho}(T)=\rho T$ and $\mathcal{M}_{r, \rho}(T)=T\rho$. We note that since $\psi_{\KMB}$ is symmetric in the sense that $\psi_{\KMB}(x, y) = \psi_{\KMB}(y, x)$, the superoperator $\pi_\rho^{\psi_\KMB}$ is self-adjoint with respect to the Hilbert--Schmidt inner product and~\eqref{eq:KMB-inner-nr} holds.

Now to derive an equivalent expression for the KMB~Fisher information, we first note that (see, e.g.,~\cite{bardet2017estimating})
$$[\ac_j, \log \rho] = \pi_\rho^{\phi_{\KMB}} \big([\ac_j, \rho]\big),$$
where 
$$\phi_{\KMB}(x, y) = \frac{1}{\psi_{\KMB}(x, y)} = \frac{\log x - \log y}{x-y}.$$
This equation motivates the definition of the \emph{KMB score operator} as
$$S_{\rho, j}^{\KMB} := \pi_\rho^{\phi_{\KMB}} ([\ac_j, \rho])=[\ac_j, \log \rho].$$
On the other hand, $\pi_\rho^{\psi_\KMB}\circ \pi_\rho^{\phi_\KMB}= \pi_\rho^{\psi_\KMB\cdot \phi_\KMB} =\pi_\rho^1 $ is the identity superoperator, which implies $[\ac_j, \rho] = \pi_\rho^{\psi_\KMB}\big(S_{\rho, j}^{\KMB}\big)$. Putting these together, we find that
\begin{align*}
    \tr \Big( \rho \big[ \ac_j^\dagger , [\ac_j, \log \rho ]\big] \Big) & = \tr \Big( [\ac_j, \rho]^\dagger\cdot   [\ac_j, \log \rho ]\Big) = \tr \Big( [\ac_j, \rho]^\dagger\cdot  S_{\rho, j}^{\KMB}\Big) = \Big\| S_{\rho, j}^\KMB \Big\|_{\rho, \KMB}^2,
\end{align*}
where $\| T \|_{\rho, \KMB}^2 \coloneqq \langle T,T\rangle_{\rho, \KMB}$. Then, summing over $j$, the KMB~Fisher information defined in~\eqref{eq:deBruijn}, can equivalently be written as\footnote{As shown in~\cite{beigi2023towards} the score operator $S_{\rho, j}^\KMB$ and $\| S_{\rho, j}^\KMB \|_{\rho, \KMB}^2$ are well-defined even if $\rho$ is not faithful.} 
\begin{align}\label{eq:KMB-norm-0}
	I_\KMB (\rho) = \sum_{j=1}^m \Big\| S_{\rho, j}^\KMB \Big\|_{\rho, \KMB}^2.
\end{align}

\medskip

The KMB inner product can be challenging to work with, especially when we are interested in applying algebraic methods. A significant problem with this inner product is that, unlike in the classical setting, it is not linear with respect to $\rho$.\footnote{Roughly speaking, in the classical case, only the values of $\psi_{\KMB}(x,y)$ at points $x=y$ matter, and taking the limit of $y\to x$ we find that $\psi_{\KMB}(x, x) = x$. Then, $\pi_\rho^{\KMB}(R) = \rho R$ and this is why the inner product $\langle \cdot, \cdot\rangle_{\rho, \KMB}$ is linear in $\rho$ in the commutative case.} However, in contrast to the classical context, there are plenty of options for defining Fisher information in the quantum or non-commutative settings~\cite{LesniewskiRuskai99}. Therefore, we can explore alternative quantum Fisher information metrics to work with. 

According to~\cite[Proposition 2.1]{hiai2016contraction}, there exists a unique probability measure $\omega$ on $[0,1]$ such that 
\begin{align}\label{eq:phi-kmb-int-00}
\phi_\KMB (x,y) = \int_{0}^1 \Big( \frac{1}{x+ty} + \frac{1}{tx+y}\Big) \frac{t+1}{2} \dd \omega(t).
\end{align}
In fact, it can be verified by direct integration that this equation holds for $\dd \omega(t) = \frac{2}{(1+t)^2}\dd t$.
Motivated by this equation,
for any $t \in [0,1]$, we define 
\begin{align}\label{eq:def-psi-kt-x,y}
\psi_{1,t}(x, y) = \frac{x+ty}{1+t},\qquad  \psi_{2,t}(x,y) = \frac{tx+y}{1+t},
\end{align}
and let 
$$\phi_{k,t}(x, y) = \frac{1}{\psi_{k, t}(x, y)}, \qquad  k=1, 2.$$
Then, by~\eqref{eq:phi-kmb-int-00} we have
\begin{equation}\label{KMB-Simple}
\pi_\rho^{\phi_\KMB} = \frac 12 \int_{0}^1 \big( \pi_\rho^{\phi_{1,t}} + \pi_\rho^{\phi_{2,t}} \big) \dd \omega(t).
\end{equation}
We can also define inner products using functions $\psi_{1, t}(x, y)$ and $\psi_{2,t}(x, y)$ as
\begin{equation}\label{eq:LinearInner1}
    \langle T,R\rangle_{\rho, 1, t} := \tr\big(  \pi_\rho^{\psi_{1,t}}(T)^{\dagger} R\big) 
    = \frac{1}{1+t}\Big(\tr\big( T^\dagger \rho R  \big) + t\,\tr\big( \rho T^\dagger  R  \big) \Big),
\end{equation}
and
\begin{equation}\label{eq:LinearInner2}
    \langle T, R\rangle_{\rho,2, t} := \tr\big( \pi_\rho^{\psi_{2, t}}(T)^{\dagger} R\big) 
= \frac{1}{1+t}\Big(t\,\tr\big(  T^\dagger \rho  R  \big) + \tr\big( \rho T^\dagger  R  \big) \Big).
\end{equation}
The main interesting property of these inner products is that they are linear with respect to $\rho$, on the contrary to the KMB inner product, making them much easier to work with. 

Continuing the above framework, we can also define the score operators with respect to $\phi_{1,t}, \phi_{2,t}$ as
$$ S_{\rho, j}^{1,t} = \pi_\rho^{\phi_{1, t}} ([\ac_j, \rho]), \quad\quad  S_{\rho, j}^{2,t} = \pi_\rho^{\phi_{2, t}} ([\ac_j, \rho]).$$ 
Then, using the integral representation~\eqref{KMB-Simple}, we write
\begin{align}
I_\KMB(\rho) &= 
\sum_{j=1}^{m} \tr \Big( \pi_\rho^{\phi}\big([\ac_j, \rho]\big)^\dagger [\ac_j, \rho]\Big) \nonumber \\
&= \frac 12 \sum_{j=1}^{m} \int_{0}^1 \bigg( \tr \Big( \pi_\rho^{\phi_{1, t}}\big([\ac_j, \rho]\big)^\dagger [\ac_j, \rho] \Big) +  \tr \Big( \pi_\rho^{\phi_{2, t}}\big([\ac_j, \rho]\big)^\dagger [\ac_j, \rho] \Big) \bigg) \dd \omega(t) \nonumber \\
&= \frac 12 \sum_{j=1}^{m} \int_{0}^1 \Big(\, \Big\| S_{\rho, j}^{1,t} \Big\|_{\rho,1, t}^2 + \Big\| S_{\rho, j}^{2,t} \Big\|_{\rho,2, t}^2 \Big) \dd \omega(t) \label{eq:HiaRuskai}.
\end{align}
Equation~\eqref{eq:HiaRuskai} allows us to work with the KMB~Fisher information by looking at the linear inner products $\langle \cdot ,\cdot\rangle_{\rho,1, t}$ and $\langle \cdot, \cdot\rangle_{\rho,2, t}$.

We note that $I_\KMB(\rho) $ is not necessarily finite. However, by the following lemma $\big\| S_{\rho, j}^{k,t} \big\|_{\rho,k, t}$ is finite for any $t\in (0,1]$ and $k=1, 2$.

\begin{lemma}\label{lem:Fisher-finite}
Let $\rho$ be an $m$-mode quantum state with finite second-order moments. Then $\big\| S_{\rho, j}^{k,t} \big\|_{\rho,k, t}$ is finite for any $t\in (0,1]$, $k=1, 2$ and $1\leq j\leq m$.
\end{lemma}

\begin{proof}
We prove the lemma only for $k=1$ as the case of $k=2$ is similar. To this end, we take the same approach as in the proof of \cite[Lemma 4]{beigi2023towards}. We first compute
\begin{align*}
\Big\| S_{\rho, j}^{1,t} \Big\|_{\rho,1, t}^2 & = \tr \Big( [\ac_j, \rho]^\dagger\cdot \pi_\rho^{\phi_{1, t}} ([\ac_j, \rho])  \Big)\\
& = \tr \Big( [\rho, \ac_j^\dagger ] \cdot \pi_\rho^{\phi_{1, t}} ([\ac_j, \rho])\Big) \\
& = (1+t) \tr \Big( \pi_\rho^{x-y} (\ac_j^\dagger) \cdot \pi_\rho^{\frac{1}{x+ty}} \circ \pi_\rho^{y-x}  (\ac_j)\Big) \\
& = (1+t) \tr \Big( \ac_j^\dagger \cdot \pi_\rho^{y-x} \circ \pi_\rho^{\frac{1}{x+ty}} \circ \pi_\rho^{y-x}  (\ac_j)\Big) \\
& = (1+t) \tr \Big( \ac_j^\dagger \cdot  \pi_\rho^{\frac{(x-y)^2}{x+ty}}  (\ac_j)\Big).
\end{align*}
Here, we use $\pi_\rho^{f}\circ \pi_\rho^{g} = \pi_\rho^{fg}$ which is immediate from the definition~\eqref{eq:pi-rho-f}. Next, we use~\cite[Lemma 4.1]{bardet2017estimating} and $\frac{(x-y)^2}{x+ty} \leq x+\frac 1t y$ for $x, y\geq 0$ to conclude that 
\begin{align*}
\Big\| S_{\rho, j}^{1,t} \Big\|_{\rho,1, t}^2 & \leq (1+t) \tr \Big( \ac_j^\dagger \cdot  \pi_\rho^{x+\frac 1t y}  (\ac_j)\Big) = (1+t)\Big( \tr ( \ac_j^\dagger   \rho \ac_j) +\frac 1t \tr (     \ac_j\rho\ac_j^\dagger)\Big).
\end{align*}
This gives $\big\| S_{\rho, j}^{1,t} \big\|_{\rho,1, t}^2<+\infty$ provided that both $\tr ( \ac_j^\dagger   \rho \ac_j)$ and $\tr (     \ac_j\rho\ac_j^\dagger)$ are finite, and $t >0$. 
\end{proof}

We note that this lemma and~\eqref{eq:HiaRuskai} do not necessarily imply that $I_\KMB(\rho)$ is finite since the above lemma does not give a uniform bound on $\big\| S_{\rho, j}^{k,t} \big\|_{\rho,k, t}^2$. In fact, the bound on $\big\| S_{\rho, j}^{1,t} \big\|_{\rho,k, t}^2$ takes the form $O(t^{-1})$ whose integral against $\dd \omega(t)$ is not finite. Nevertheless, as will see later, finiteness of $I_\KMB(\rho)$ is not really necessary in our arguments.

In the following, we often drop the indices $k, t$ and use the notations $\langle \cdot, \cdot\rangle_{\rho}$ and $S_{\rho, j}$ to refer to any of the linear inner products $\langle \cdot, \cdot\rangle_{\rho, k, t}$ and their corresponding score operator $S_{\rho, j}^{k,t}$, for $k=1, 2$, $t\in (0,1]$, respectively.

\subsection{Generalized symmetric lifting map}\label{Sec:LiftingMap}
In this section, we introduce the generalized symmetric lifting map first presented in~\cite{beigi2023towards}. The primary motivation behind this map is to extend the definition of symmetric functions of the form $ \widetilde g(x_1, \dots, x_n)=g\big( \frac{x_1 + \cdots + x_n}{\sqrt{n}}\big)$ associated with a given function $g(\cdot)$, to the quantum case. The symmetric lifting map introduced in~\cite{beigi2023towards} is used to relate the score operator of the state $ \rho^{\boxplus [n]} $ to the score operators of the individual $\rho_k$'s. In this paper, we further generalize the notion of the symmetric lifting map for arbitrary subsets $ \ve \subseteq [n]$.

In the following, we assume that $\rho_1, \dots, \rho_n$ are $m$-mode bosonic states and $X_1, \dots, X_{n'}$ are independent $\C^m$-valued random variables. We use the notations $\rho^{\otimes [n]} = \rho_1\otimes \cdots \otimes \rho_n$
and $X^{[n']}=(X_1, \dots, X_{n'})$ for simplicity.

For a subset $\ve \subseteq [n]$ and $z\in \C^m$, let $\cW_\ve(z)$  be the unitary operator acting on $\mathcal{H}_m^{\otimes n}$ given by 
$$ \cW_\ve(z) = \Big(\bigotimes_{k\in \ve}  D_{\frac{z}{\sqrt{|\ve|}}}\Big)\bigotimes \Big(\bigotimes_{k\in \ve^c}I \Big),$$
where as before $\mathcal{H}_m = L^2(\R^m)$ is the Hilbert space corresponding to an $m$-mode bosonic quantum system.
That is, $\cW_{\ve}(z)$ acts as $D_{z/\sqrt{|\ve|}}$ on subsystems with indices in $\ve$, and as identity elsewhere. Also, for a pair of subsets $\ve\subseteq [n], \w \subseteq [n^\prime]$, let $\cF_{\ve, \w}:\mathbb{C}^{m}\times (\C^{m})^{n^\prime}\to \mathbb{C}$ be the function given by 
$$ \cF_{\ve, \w}(z, x_1, \dots ,x_{n^\prime}) = \prod_{\ell\in \w}  D_c\bigg(\frac{z}{\sqrt{|\ve|}} , x_\ell\bigg),$$where $D_c(.,.)$ is defined in~\eqref{eq:chi-classic-D}.
Now, for an $m$-mode trace class operator $T$, we define its $(\ve,\w)$-symmetric lifting $\widetilde T_{\ve, \w}$, as a function from $(\C^{m})^{{n^\prime}}$ to bounded operators acting on $\mathcal{H}_m^{\otimes n} $ by\footnote{Hereafter, abusing the notation, we treat $[n]$ and $[n']$ as disjoint subsets indexing quantum and classical registers, respectively. Thus, subsets $\ve\subseteq [n]$ and $ \w\subseteq[n']$ are also disjoint.}
\[
	\widetilde T_{\ve, \w} (x_1, \dots, x_{n^\prime}) :  = \frac{1}{\pi^m} \int \chi_T(z)  \cF_{\ve, \w}(-z, x_1, \dots ,x_{n^\prime}) \cW_\ve(-z) \dd^{2m} z.
\]
We note that, for any $(x_1, \dots, x_{n'})$ the map $D_z \mapsto \cF_{\ve, \w}(z, x_1, \dots ,x_{n^\prime}) \cW_{\ve}(z)$ is a representation of the Weyl--Heisenberg group. Thus, using the Stone--von Neumann theorem~\cite{Hall-Book}, the operator $\widetilde T_{\ve, \w}(x_1, \dots, x_{n'})$ can be defined even if $T$ is not trace class. For more details we refer to~\cite{beigi2023towards}.

We now consider the space of measurable maps from $(\C^{m})^{{n^\prime}}$ to operators acting on $\mathcal{H}_m^{\otimes n}$ and define an inner product on this space. Letting $A, B$ be two elements in this space, we define 
\[
	\big\langle A , B \big\rangle_{\rho^{\otimes [n]}, X^{[n']}} = \mathbb{E}_{X^{[n']}} \Big[ \Big\langle A(x_1, \dots, x_{n^\prime}), B(x_1, \dots, x_{n^\prime}) \Big\rangle_{\rho^{\otimes [n]}} \Big],
\]
Here, $\langle \cdot , \cdot \rangle_{\rho^{\otimes [n]}}$ denotes any of the linear inner products defined in Subsection~\ref{Sec:KMBInner}. We, of course, assume that $A, B$ are such that the above expectation is meaningful and finite. Letting  $\|\cdot\|_{\rho^{\otimes[ n]}, X^{[n']}}$ be the norm induced by the above inner product,
we denote by $\mathfrak A_{(n, n')}=\mathfrak A_{\rho^{\otimes [n]}, X^{[n']}}$ the space of measurable maps $A$ from $(\C^{m})^{ {n^\prime}}$ to operators acting on $\mathcal{H}_m^{\otimes n}$ satisfying $\|A\|_{\rho^{\otimes [n]}, X^{[n']}} < +\infty$. We observe that $\mathfrak A_{(n, n')}$ is a Hilbert space.

To get an intuition about the symmetric lifting map, let us consider the special case where $n'=0$. In this case, the symmetric lifting map for $\ve=[n]$ is written as
\[
	\widetilde T_{[n]}   = \frac{1}{\pi^m} \int \chi_T(z)  \bigotimes_{k\in [n]}  D_{-\frac{z}{\sqrt n}} \dd^{2m} z ,
\]
and matches the one introduced in~\cite{beigi2023towards}. As shown in~\cite[Proposition 3]{beigi2023towards} we have $\tr\big(\rho^{\otimes [n]}  \widetilde T_{[n]} \big) = \tr\big(\rho^{\boxplus [n]}   T \big)$. Writing this equation with respect to the Hilbert--Schmidt inner product, we find that 
$$\big\langle\rho^{\otimes [n]} , \widetilde T_{[n]} \big\rangle = \big\langle \rho^{\boxplus [n]} ,  T\big\rangle.$$
This means that the symmetric lifting $T\mapsto \widetilde T_{[n]}$ is the adjoint of the convolution map $\rho^{\otimes [n]}\mapsto \rho^{\boxplus [n]}$ with respect to the Hilbert--Schmidt inner product. 
The following proposition is a far more generalization of this observation.

\begin{proposition}\label{Pr:InnerProduct_Symmetric}
	Let $\rho_1, \dots , \rho_{n}$ be $m$-mode bosonic quantum states with finite second moments, and let $X_1, \dots, X_{n^\prime}$ be $\C^{m}$-valued independent random vectors with finite variance. Let $\varnothing\neq \ve \subseteq [n]$ and $\w \subseteq [{n^\prime}]$ be arbitrary subsets. Then, for $m$-mode operators $T, R$ satisfying $\|T \|_{\rho^{\boxplus [n]} \star_{\frac{n'}n} X^{\boxplus [{n^\prime}]}}<+\infty$ and $\|R \|_{\rho^{\boxplus [n]} \star_{\frac{n'}n} X^{\boxplus [{n^\prime}]}}<+\infty$ norm, we have
\begin{align*}
\Big\langle \widetilde T_{\ve, \w},  \widetilde R_{[n], [{n^\prime}]} \Big\rangle_{\rho^{\otimes [n]}, X^{[n']}}
 = \tr \Big[ \big( \pi_{\rho^{\boxplus \ve} \star_{\frac{|\w|}{|\ve|}} X^{\boxplus \w}}^{\psi} (T)^\dagger \boxplus_{\frac{|\ve|}{n}} (\rho^{\boxplus {\ve^c}} \star_{\frac{|\w^c|}{|\ve^c|}} X^{\boxplus \w^c}) \big) R \Big].
\end{align*}
Here, $\psi$ is any of the functions $\psi_{k, t}$ for $k=1, 2$ and $t\in [0 ,1]$ defined in~\eqref{eq:def-psi-kt-x,y} based on which the inner product on the left hand side is defined. 
Moreover, $\ve^c$ is the complement of $\ve$ in $[n]$ and $\w^c$ is the complement of $\w$ in $[n']$.
\end{proposition}

Letting $\ve=[n]$ and $\w=[n']$, the above proposition implies that
$$\Big\langle \widetilde T_{[n], [n']},  \widetilde R_{[n], [{n^\prime}]} \Big\rangle_{\rho^{\otimes [n]}, X^{[n']}} = \big\langle  T,  R \big\rangle_{\rho^{\boxplus [n]}\star_{\frac{n'}{n}} X^{\boxplus[n']}}.$$
This means that $T\mapsto \widetilde T_{[n], [n']}$, and in general $T\mapsto \widetilde T_{\ve, \w}$ for any $\ve, \w$, are isometries.

\begin{proof}
As shown in Appendix~\ref{app:trace-dense}, trace class operators are dense in the linear space of operators equipped with norm $\|\cdot \|_{\rho^{\boxplus [n]} \star X^{\boxplus [{n^\prime}]}}$. Thus, we may assume that $R, T$ are trace class. Also, by linearity it suffices to prove this equation when $\psi(x, y)$ is either $\psi_{1, 0}(x, y)=x$ or $\psi_{2, 0}(x, y)=y$ as in general any $\psi_{k, t}(x, y)$ is a linear combination of these two. For simplicity and clarity of presentation, we assume that $\psi=\psi_{1, 0}$. The same argument works in general. 
Starting with the left hand side, we can write
\begin{align*}
&\Big\langle \widetilde T_{\ve, \w},  \widetilde R_{[n], [{n^\prime}]} \Big\rangle_{\rho^{\otimes [n]}, X^{[n']}} \\
& = \mathbb{E}_{X^{[n^\prime]}} \bigg[ \Big\langle \widetilde T_{\ve, \w} (x_1, \dots, x_{n^\prime}), \widetilde R_{[n], [{n^\prime}]} (x_1, \dots, x_{n^\prime}) \Big\rangle_{\rho^{\otimes [n]}} \bigg] \\
&= \mathbb{E}_{X^{[n^\prime]}} \bigg[ \frac{1}{\pi^{2m}} \int \overline{\chi_T(z)}\,  \chi_R(z^\prime)\, \overline{\cF_{\ve, \w}(-z, x_1, \dots ,x_{n^\prime}) }\,\cF_{[n], [{n^\prime}]}(-z^\prime, x_1, \dots ,x_{n^\prime}) \\ & \qquad \qquad \qquad \times \Big\langle \cW_\ve(-z), \cW_{[n]}(-z^\prime) \Big\rangle_{\rho^{\otimes [n]}} \dd^{2m} z \dd^{2m} z^\prime \bigg]\\
&= \frac{1}{\pi^{2m}} \int \overline{\chi_T(z)} \, \chi_R(z^\prime) \, \mathbb{E}_{X^{[n^\prime]}} \Big[ \overline{\cF_{\ve, \w}(-z, x_1, \dots ,x_{n^\prime})} \, \cF_{[n], [{n^\prime}]}(-z^\prime, x_1, \dots ,x_{n^\prime})\Big] \\ 
& \qquad \qquad \qquad \times \Big\langle \cW_\ve(-z), \cW_{[n]}(-z^\prime) \Big\rangle_{\rho^{\otimes [n]}} \dd^{2m} z \dd^{2m} z^\prime .
\end{align*}
For the last factor, we compute

\begin{align*}
	\Big\langle \cW_\ve(-z), \cW_{[n]}(-z^\prime) \Big\rangle_{\rho^{\otimes [n]}} & =  \prod_{k\in \ve} \Big\langle D\Big(-\frac{z}{\sqrt{|\ve|}} \Big), D\Big(-\frac{z^\prime }{\sqrt n} \Big)\Big\rangle_{\rho_k} \cdot \prod_{k\in \ve^c} \Big\langle I, D\Big(-\frac{z^\prime }{\sqrt n} \Big)\Big\rangle_{\rho_k} \\
	& = \prod_{k\in \ve} \tr\Big[ \rho_k   D\Big(-\frac{z^\prime }{\sqrt n} \Big) D\Big(\frac{z}{\sqrt{|\ve|}} \Big)\Big] \cdot \prod_{k\in \ve^c} \tr\Big[ \rho_k   D\Big(-\frac{z^\prime }{\sqrt n} \Big)  \Big]\\
	& =  e^{i\theta_{z, z^\prime}} \prod_{k\in \ve} \tr\Big[ \rho_k   D\Big(\frac{z}{\sqrt{|\ve|}}-\frac{z^\prime }{\sqrt n} \Big) \Big] \cdot \prod_{k\in \ve^c} \tr\Big[ \rho_k   D\Big(-\frac{z^\prime }{\sqrt n} \Big)  \Big]\\
	& =  e^{i\theta_{z, z^\prime}} \prod_{k\in \ve} \chi_{\rho_k}   \Big(\frac{z}{\sqrt{|\ve|}}-\frac{z^\prime }{\sqrt n} \Big) \cdot \prod_{k\in \ve^c} \chi_{\rho_k} \Big(-\frac{z^\prime }{\sqrt n} \Big)  \\
	& =  e^{i\theta_{z, z^\prime}} \chi_{\rho^{\boxplus \ve}}   \Big(z-\sqrt{\frac{ |\ve|}{n}}z^\prime \Big) \cdot  \chi_{\rho^{\boxplus \ve^c}} \Big(-\sqrt{\frac{|\ve^c| }{n}}z^\prime \Big) ,
\end{align*}
where $i\theta_{z, z^\prime}=\frac12 \sqrt{\frac{|\ve|}{n}} (z^\top \bar{z}^\prime -\bar{z}^\top z^\prime)$ and in the third line we use the product rule for displacement operators~\eqref{eq:productRule}. Also, we have
\begin{align*}
&\mathbb{E}_{X^{[n^\prime]}} \Big[ \overline{\cF_{\ve, \w}(-z, x_1, \dots ,x_{n^\prime})} \, \cF_{[n], [{n^\prime}]}(-z^\prime, x_1, \dots ,x_{n^\prime})  \Big]\\
& = \prod_{\ell\in \w}  \mathbb E\Big[ D_\cc\Big(\frac{z}{\sqrt{|\ve|}} , x_\ell\Big) \, D_\cc\Big(-\frac{z^\prime }{\sqrt n}  , x_\ell \Big)\Big] \cdot \prod_{\ell \in \w^c} \mathbb E\Big[D_\cc\Big(\frac{-z^\prime}{\sqrt n} , x_\ell\Big)\Big]\\
& = \prod_{\ell\in \w}  \mathbb E\Big[ D_\cc\Big(\frac{z}{\sqrt{|\ve|}}-\frac{z^\prime }{\sqrt n} , x_\ell\Big) \Big] \cdot \prod_{\ell \in \w^c} \mathbb E\Big[D_\cc\Big(\frac{-z^\prime}{\sqrt n} , x_\ell\Big)\Big]\\
& = \prod_{\ell\in \w}  \chi_{X_\ell}\Big(\frac{z}{\sqrt{|\ve|}}-\frac{z^\prime }{\sqrt n} , x_\ell\Big) \cdot \prod_{\ell \in \w^c} \chi_{X_\ell}\Big(\frac{-z^\prime}{\sqrt n} , x_\ell\Big)\\
& = \chi_{X^{\boxplus \w}} \bigg(\sqrt{\frac{|\w|}{|\ve|}} z - \sqrt{\frac{|\w|}{n}} z^\prime \bigg)\cdot \chi_{X^{\boxplus \w^c }} \bigg(\ - \sqrt{\frac{|\w^c|}{n}} z^\prime \bigg).
\end{align*}
Putting all these together yields
\begin{align}\label{eq:prop1-left-hand-side-putting}
& \big\langle \widetilde T_{\ve, \w},  \widetilde R_{[n], [{n^\prime}]} \big\rangle_{\rho^{\otimes [n]}, X^{[n']}}\nonumber\\
& = \frac{1}{\pi^{2m}} \int \overline{\chi_T(z)} \, \chi_R(z^\prime) \, \chi_{X^{\boxplus \w}} \bigg(\sqrt{\frac{|\w|}{|\ve|}} z - \sqrt{\frac{|\w|}{n}} z^\prime \bigg)\cdot \chi_{X^{\boxplus \w^c}} \bigg(\ - \sqrt{\frac{|\w^c|}{n}} z^\prime \bigg)  \nonumber\\
&   \qquad \quad \qquad \times e^{i\theta_{z, z^\prime}} \chi_{\rho^{\boxplus \ve}}   \Big(z-\sqrt{\frac{ |\ve|}{n}}z^\prime \Big) \cdot  \chi_{\rho^{\boxplus \ve^c}} \Big(-\sqrt{\frac{|\ve^c| }{n}}z^\prime \Big)\dd^{2m} z \dd^{2m} z^\prime.
\end{align}

To handle the other side, we use linearity and~\eqref{eq:inverse_Char} to write
\begin{align}\label{eq:prop1-right-hand-side-f}
&\tr \Big[ \big( \pi_{\rho^{\boxplus \ve} \star_{\frac{|\w|}{|\ve|}} X^{\boxplus \w}}^{\psi} (T)^\dagger \boxplus_{\frac{|\ve|}{n}} (\rho^{\boxplus {\ve^c}} \star_{\frac{|\w^c|}{|\ve^c|}} X^{\boxplus \w^c}) \big) R \Big]\nonumber \\
&= \frac{1}{\pi^{2m}} \int \overline{\chi_T(z)} \, \chi_R(z^\prime) \,  \tr \bigg[ \Big( \pi_{\rho^{\boxplus \ve} \star_{\frac{|\w|}{|\ve|}} X^{\boxplus \w}}^{\psi} (D_{-z})^\dagger \boxplus_{\frac{|\ve|}{n}} \big(\rho^{\boxplus {\ve^c}} \star_{\frac{|\w^c|}{|\ve^{c}|}} X^{\boxplus \w^c}\big) \Big) D_{-z^\prime} \bigg] \dd^{2m} z \: \dd^{2m} z^\prime.
\end{align}
Also, we can write
\begin{align*}
& \tr \bigg[\Big( \pi_{\rho^{\boxplus \ve} \star_{\frac{|\w|}{|\ve|}} X^{\boxplus \w}}^{\psi} (D_{-z})^\dagger \boxplus_{\frac{|\ve|}{n}} \big(\rho^{\boxplus {\ve^c}} \star_{\frac{|\w^c|}{|\ve^{c}|}} X^{\boxplus \w^c}\big) \Big) D_{-z^\prime} \bigg]\\
&= \chi_{ \pi_{\rho^{\boxplus \ve} \star_{\frac{|\w|}{|\ve|}} X^{\boxplus \w}}^{\psi} (D_{-z})^\dagger \boxplus_{\frac{|\ve|}{n}} \big(\rho^{\boxplus {\ve^c}} \star_{\frac{|\w^c|}{|\ve^{c}|}} X^{\boxplus \w^c}\big) } (-z')\\
& = \chi_{\pi_{\rho^{\boxplus \ve} \star_{\frac{|\w|}{|\ve|}} X^{\boxplus \w}}^{\psi} (D_{-z})^\dagger} \bigg(-\sqrt{\frac{|\ve|}{n}} z'\bigg) \chi_{ \rho^{\boxplus {\ve^c}} \star_{\frac{|\w^c|}{|\ve^{c}|}} X^{\boxplus \w^c} } \Big(-\sqrt{\frac{|\ve^c|}{n}}z'\Big)\\
& = \tr\bigg[\rho^{\boxplus \ve} \star_{\frac{|\w|}{|\ve|}} X^{\boxplus \w} \cdot D_{-\sqrt{\frac{|\ve|}{n}} z'} D_{z}\bigg] \chi_{ \rho^{\boxplus {\ve^c}} \star_{\frac{|\w^c|}{|\ve^{c}|}} X^{\boxplus \w^c} } \Big(-\sqrt{\frac{|\ve^c|}{n}}z'\Big)\\
& = e^{i\theta_{z, z'}}\tr\bigg[\rho^{\boxplus \ve} \star_{\frac{|\w|}{|\ve|}} X^{\boxplus \w} \cdot  D_{z-\sqrt{\frac{|\ve|}{n}} z'} \bigg] \chi_{ \rho^{\boxplus {\ve^c}} \star_{\frac{|\w^c|}{|\ve^{c}|}} X^{\boxplus \w^c} } \Big(-\sqrt{\frac{|\ve^c|}{n}}z'\Big)\\
& = e^{i\theta_{z, z'}}\chi_{\rho^{\boxplus \ve} \star_{\frac{|\w|}{|\ve|}} X^{\boxplus \w}}  \bigg(z-\sqrt{\frac{|\ve|}{n}} z'\bigg) \chi_{ \rho^{\boxplus {\ve^c}} \star_{\frac{|\w^c|}{|\ve^{c}|}} X^{\boxplus \w^c} } \Big(-\sqrt{\frac{|\ve^c|}{n}}z'\Big)\\
& = e^{i\theta_{z, z'}}\chi_{\rho^{\boxplus \ve}}  \bigg(z-\sqrt{\frac{|\ve|}{n}} z'\bigg) \chi_{X^{\boxplus \w}}\bigg(\sqrt{\frac{|\w|}{|\ve|}}z-\sqrt{\frac{|\w|}{n}} z'\bigg)  \chi_{ \rho^{\boxplus {\ve^c}} } \Big(-\sqrt{\frac{|\ve^c|}{n}}z'\Big) \chi_{X^{\boxplus \w^c}}\bigg(-\sqrt{\frac{|\w^c|}{n}}z'\bigg).
\end{align*}
Here, we use the definition of the characteristic function, as well as~\eqref{eq:convolution-char} and~\eqref{eq:char-q-c-conv}. Also, in the fifth line we apply~\eqref{eq:productRule}. Using this in~\eqref{eq:prop1-right-hand-side-f} and comparing to~\eqref{eq:prop1-left-hand-side-putting}, we obtain the desired identity.
\end{proof}

We can now state the following important lemma which can be interpreted as the quantum generalization of~\cite[Lemma 1]{madiman2007generalized}.

\begin{lemma}\label{NiceInner}
Let $\rho_1, \dots , \rho_{n}$ be $m$-mode bosonic quantum states with finite second moments, and let $X_1, \dots, X_{n^\prime}$ be $\C^{m}$-valued independent random vectors with finite variance. Let $R$ be an operator satisfying $\| R\|_{\rho^{\boxplus [n]} \star_{\frac{n'}{n}} X^{\boxplus [{n^\prime}]}} < +\infty$. Also, for any subsets $\varnothing\neq \ve \subseteq [n]$ and $\w \subseteq [{n^\prime}]$, let $S_{\rho^{\boxplus \ve} \star_{\frac{|\w|}{|\ve|}} X^{\boxplus \w}, j}$ be the  score operator of $\rho^{\boxplus \ve} \star_{\frac{|\w|}{|\ve|}} X^{\boxplus\w}$ on the $j$-th mode with respect to a linear inner product which we fix in advance. Then, we have
\[
	\Big\langle \widetilde S_{\rho^{\boxplus \ve} \star_{\frac{|\w|}{|\ve|}} X^{\boxplus \w}, j}, \widetilde R_{[n], [{n^\prime}]} \Big\rangle_{\rho^{\otimes [n]}, X^{[n^\prime]} } = \sqrt{\frac{|\ve|}{n}} \Big\langle S_{\rho^{\boxplus [n]} \star_{\frac{n'}{n}} X^{\boxplus [{n^\prime}]}, j}, R \Big\rangle_{\rho^{\boxplus [n]} \star_{\frac{n'}{n}} X^{\boxplus [{n^\prime}]} },
\]
where $\widetilde S_{\rho^{\boxplus \ve} \star_{\frac{|\w|}{|\ve|}} X^{\boxplus \w}, j} =\widetilde{(S_{\rho^{\boxplus \ve} \star_{\frac{|\w|}{|\ve|}} X^{\boxplus \w}, j})_{\ve, \w}}$ denotes the $(\ve, \w)$-symmetric lifting of $ S_{\rho^{\boxplus \ve} \star_{\frac{|\w|}{|\ve|}} X^{\boxplus \w}, j}$.
\end{lemma}

Since $\rho_k$'s and $X_\ell$'s have finite second-order moments, by Lemma~\ref{lem:Fisher-finite}, $\big\| S_{\rho^{\boxplus \ve} \star_{\frac{|\w|}{|\ve|}} X^{\boxplus \w}, j} \big\|_{\rho^{\boxplus \ve} \star_{\frac{|\w|}{|\ve|}} X^{\boxplus \w}}$ is indeed finite and the above equation makes sense.

\begin{proof}
Using Proposition~\ref{Pr:InnerProduct_Symmetric} we compute
\begin{align*}
\Big\langle \widetilde S_{\rho^{\boxplus \ve} \star_{\frac{|\w|}{|\ve|}} X^{\boxplus \w}, j}, \widetilde R_{[n], [{n^\prime}]} \Big\rangle_{\rho^{\otimes [n]}, X^{[n^\prime]} } &= \tr \Big( \big( \pi_{\rho^{\boxplus \ve} \star_{\frac{|\w|}{|\ve|}} X^{\boxplus \w}}^{\psi} (S_{\rho^{\boxplus \ve} \star_{\frac{|\w|}{|\ve|}} X^{\boxplus \w}, j})^\dagger \boxplus_{\frac{|\ve|}{n}} (\rho^{\boxplus {\ve^c}} \star_{\frac{|\w^c|}{|\ve^c|}} X^{\boxplus \w^c}) \big) R \Big) \\
&= \tr \Big( \big( [\ac_j , \rho^{\boxplus \ve} \star_{\frac{|\w|}{|\ve|}} X^{\boxplus \w} ]^\dagger \boxplus_{\frac{|\ve|}n} (\rho^{\boxplus {\ve^c}} \star_{\frac{|\w^c|}{|\ve^c|}} X^{\boxplus \w^c}) \big) R \Big) \\
&= \sqrt{\frac{|\ve|}{n}} \tr \Big( [\ac_j , \rho^{\boxplus [n]} \star_{\frac{n'}{n}} X^{\boxplus [{n^\prime}]}]^\dagger R \Big) \\
&= \sqrt{\frac{|\ve|}{n}}  \Big\langle S_{\rho^{\boxplus [n]} \star_{\frac{n'}{n}} X^{\boxplus [{n^\prime}]}, j}, R \Big\rangle_{\rho^{\boxplus [n]} \star_{\frac{n'}{n}} X^{\boxplus [{n^\prime}]} } ,
\end{align*}
where in the third line we use the definition of the score operator and~\eqref{eq:Comm_Conv}.
\end{proof}

And here is an important corollary of this lemma.

\begin{corollary}\label{corollary:projection-score-square}
Fix an inner product as above and 
let $\Pi$ be the orthogonal projection on the closure of the linear space $\left\{ \widetilde R_{[n], [{n^\prime}]}:\,   \| R \|_{\rho^{\boxplus [n]} \star_{\frac{n'}{n}} X^{\boxplus [{n^\prime}]}} < +\infty   \right\}$. Then, for any subsets $\ve\subseteq [n]$ and $\w\subseteq [n']$ we have
\begin{equation}\label{eq:projection-score-square}
 \sqrt{\frac{n}{|\ve|}}  \Pi\big(\widetilde S_{\rho^{\boxplus \ve} \star_{\frac{|\w|}{|\ve|}} X^{\boxplus \w}, j} \big) =   \widetilde S_{\rho^{\boxplus [n]} \star_{\frac{n'}{n}} X^{\boxplus [{n^\prime}]}, j}.
\end{equation} 
\end{corollary}

\begin{proof}
By applying Lemma~\ref{NiceInner} twice, we obtain
\begin{align*}
	\Big\langle \widetilde S_{\rho^{\boxplus \ve} \star_{\frac{|\w|}{|\ve|}} X^{\boxplus \w}, j}, \widetilde R_{[n], [{n^\prime}]} \Big\rangle_{\rho^{\otimes [n]}, X^{[n^\prime]} } & = \sqrt{\frac{|\ve|}{n}} \Big\langle S_{\rho^{\boxplus [n]} \star_{\frac{n'}{n}} X^{\boxplus [{n^\prime}]}, j}, R \Big\rangle_{\rho^{\boxplus [n]} \star_{\frac{n'}{n}} X^{\boxplus [{n^\prime}]} }\\
	& = \sqrt{\frac{|\ve|}{n}} \Big\langle \widetilde S_{\rho^{\boxplus [n]} \star_{\frac{n'}{n}} X^{\boxplus [{n^\prime}]}, j}, \widetilde R_{[n], [n']} \Big\rangle_{\rho^{\otimes [n]}, X^{[n^\prime]} }.
\end{align*}
Since this identity holds for any $R$, the desired identity~\eqref{eq:projection-score-square} holds. 
\end{proof}

\subsection{Decomposition of tensor product spaces}\label{Sec:Decomp}
Let $\langle \cdot, \cdot \rangle_\rho$ be any of the linear inner products $\langle \cdot, \cdot\rangle_{\rho, k, t}$, for $k=1, 2$ and $t\in (0,1]$, with the corresponding functions $\psi, \phi$, defined in~\eqref{eq:LinearInner1} and~\eqref{eq:LinearInner2}. We note that for any such inner product and all operators $T_A, T_B$ and quantum states $\rho_A, \rho_B$ on subsystems $A, B$ respectively, we have 
\begin{equation}\label{eq:inner-assumptionIdentity}
\PiInner{\rho_A \otimes \rho_B}{\psi}(T_A \otimes {I}_B) = \PiInner{\rho_A}{\psi}(T_A) \otimes \rho_B, \quad \PiInner{\rho_A \otimes \rho_B}{\psi}({I}_A \otimes T_B) = \rho_A \otimes \PiInner{\rho_B}{\psi}(T_B).
\end{equation}
We note that this equation does not hold for the KMB~inner product, but because of linearity it holds for the other inner products we consider here.

Let $\rho_1, \dots, \rho_{n}$ be $m$-mode bosonic quantum states and $X_1, \dots, X_{n^\prime}$ be $\C^{m}$-valued, independent random vectors. Recall that $\mathfrak A_{(n, n')}$ is the space of maps $A$ from $(\C^{m})^{{n^\prime}}$ to operators acting on $\mathcal{H}_m^{\otimes n}$ satisfying $\|A\|_{\rho^{\otimes [n]}, X^{[n^\prime]}}<+\infty$. For any $1\leq k \leq n$, define the superoperator $\E_k : \mathfrak A_{(n, n')} \to \mathfrak A_{(n, n')}$ by
 $$\E_k A (x_1, \dots, x_{n^\prime}) := \tr_k\bigg(\rho_k A(x_1, \dots, x_{n^\prime})\bigg) \otimes {I}_k,$$ 
where $\tr_k$ stands for the partial trace over the $k$-th subsystem and $I_k$ is the identity operator acting on the $k$-th subsystem. Also, for $1\leq \ell \leq n'$ we define the superoperators $\E_\ell^c : \mathfrak A_{(n, n')} \mapsto \mathfrak A_{(n, n')}$ as
$$\E_\ell^c A  := \mathbb{E}_{X_\ell}[A].$$ 
We note that $\E_k$ and $\E_\ell^c$ are projections. Moreover, using~\eqref{eq:inner-assumptionIdentity} it can be verified that they are self-adjoint with respect to $\langle \cdot, \cdot\rangle_{\rho^{\otimes [n]}, X^{[n^\prime]}}$, and these projections mutually commute by definition. Therefore, denoting the identity superoperator by $\mathcal I$, we can write 
\begin{equation}\label{VarianceDrop}
\mathcal I =  \prod_{k=1}^{n} \big(\E_k + (\mathcal{I} - \E_k)\big) \cdot \prod_{\ell=1}^{{n^\prime}} \big(\E_\ell^c + (\mathcal{I} - \E_\ell^c) \big)   = \sum_{\ve \subseteq [n], \w \subseteq [{n^\prime}]} {\mathcal P}_{\ve,\w} ,
\end{equation}
where 
$${\mathcal P}_{\ve, \w} = \Big(\prod_{k\notin \ve} \E_k \Big) \prod_{k\in \ve} (\mathcal{I} - \E_k) \cdot \Big(\prod_{\ell\notin \w} \E_\ell^c \Big) \prod_{\ell\in \w} (\mathcal{I} - \E_\ell^c).$$ 
We observe that ${\mathcal P}_{\ve,\w}$ is an orthogonal projection for any $(\ve , \w)$, and ${\mathcal P}_{\ve,\w}{\mathcal P}_{\ve^\prime, \w^\prime}=0$ if $\ve \neq \ve^\prime$ or $\w \neq \w^\prime$. The point is that if there exists $k \in \ve$ such that $k \notin \ve^\prime$, then we have a factor $\E_k \times (\mathcal{I} - \E_k) = 0$ in ${\mathcal P}_{\ve,\w}{\mathcal P}_{\ve^\prime, \w^\prime}$. The same story goes with $\w, \w^\prime$. Thus, the images of ${\mathcal P}_{\ve, \w}$, for $\ve\subseteq [n]$ and $\w \subseteq [{n^\prime}]$, form a decomposition of the Hilbert space $\mathfrak{A}_{(n,{n^\prime})}$ corresponding to the inner product $\langle \cdot, \cdot \rangle_{\rho^{\otimes [n]}, X^{[n^\prime]}}$, into a direct sum of orthogonal subspaces.

\section{Proof of the main results}\label{Sec:proof_mainres}

In this section, we use the tools developed in the previous section to prove our main results. We first state generalizations of our main results which include both classical and quantum registers.

\begin{theorem}\label{Th:Quantum_Classical_Shannon}
Let $\rho_1, \dots, \rho_{n}$ be $m$-mode bosonic quantum states with finite second moments, and $X_1, \dots , X_{n^\prime}$ be $\C^{m}$-valued, independent random vectors with finite variance. Let $\mathcal{C}$ be a collection of pairs $(\ve,\w)$ consisting of subsets $\varnothing\neq \ve \subseteq [n]$ and $\w \subseteq [{n^\prime}]$. Let $r$ be the maximum number of times an index from $[n]$ or $[{n^\prime}]$ appears in elements in $\mathcal{C}$. 
Then, we have
\begin{equation}\label{eq:Generalized_quantum_Classical_EPI}
\exp \bigg( \frac 1m  S\big(\rho^{\boxplus [n]} \star_{\frac{n'}{n}} X^{\boxplus [{n^\prime}]}\big) \bigg) \geq \frac{1}{ rn } \sum_{(\ve, \w) \in \mathcal{C}} |\ve|  \exp \bigg( \frac 1m  S\big(\rho^{\boxplus \ve} \star_{\frac{|\w|}{\ve}} X^{\boxplus \w}\big) \bigg).
\end{equation}
\end{theorem}

Here is the generalization of Theorem~\ref{Th:Shannon_Stam_KMB-intro} including classical random variables.

\begin{theorem}\label{Th:Shannon_Stam_KMB}
Let $\rho_1, \dots, \rho_{n}$ be $m$-mode bosonic quantum states with finite second moments, and $X_1, \dots , X_{n^\prime}$ be $\C^{m}$-valued, independent random vectors with finite variance. Let $\mathcal{C}$ be a collection of pairs $(\ve, \w)$ of subsets $\varnothing\neq\ve \subseteq [n]$ and $\w \subseteq [{n^\prime}]$. Let $r$ be the maximum number of times an index from $[n]$ or $[{n^\prime}]$ appears in elements in $\mathcal{C}$. Then, for any probability distribution $\mu$ on $\mathcal{C}$ we have
\begin{align}\label{eq:Shannon-Stam_KMB}
	I_\KMB\Big(\rho^{\boxplus [n]} \star_{\frac{n'}{n}} X^{\boxplus [{n^\prime}]}\Big) \leq r   n \sum_{(\ve, \w) \in \mathcal{C}} \frac{1}{| \ve |} \, \mu_{(\ve,\w)}^2 \, I_\KMB\Big(\rho^{\boxplus \ve} \star_{\frac{|\w|}{|\ve|}} X^{\boxplus \w}\Big).
\end{align}
\end{theorem}

We note that in the above theorems $n'$, the number of classical registers, can be equal to zero. This is why these theorems are generalizations of Theorem~\ref{Th:Quantum_Shannon} and Theorem~\ref{Th:Shannon_Stam_KMB-intro}. Moreover, inequality~\eqref{eq:Shannon-Stam_KMB} means that if the left hand side is infinite, then the right hand side is also infinite.

In the following, we first give the proof of Theorem~\ref{Th:Quantum_Classical_Shannon} assuming Theorem~\ref{Th:Shannon_Stam_KMB}, and then move on to the proof of Theorem~\ref{Th:Shannon_Stam_KMB}.

\subsection{Proof of Theorem~\ref{Th:Quantum_Classical_Shannon}}

For simplicity of presentation and conveying the main ideas, we first assume that $n'=0$ and there is no classical random variable, i.e., we first give the proof of Theorem~\ref{Th:Quantum_Shannon}. Later, we will discuss the more general case.

Let $\mu$ be a distribution on $\mathcal C$ satisfying
$$r  \mu_{\ve, \w}\leq 1, \qquad \forall (\ve, \w)\in \mathcal C.$$
Then, we can use Lemma~\eqref{lem:subsets-linear}  in Appendix~\ref{app:design} to find some $n'' \geq 1$ and a collection of subsets $\w'_{(\ve, \w)}\subseteq [n'']$ as well as real numbers $h_\ell$ for every $\ell\in [n'']$, such that the following hold:
\begin{itemize}
\item each element of $[n'']$ appears in at most $r$ subsets $\w'_{(\ve, \w)}$
\item $\sum_{\ell \in [n'']} h_{\ell} = n$
\item $\sum_{\ell \in \w'_{\ve, \w}}  h_\ell = r\mu_{\ve, \w} n.$
\end{itemize}
Now to any $\ell \in [n'']$ we associate a $\C^{2m}$-valued independent centered \emph{Gaussian} random variable $Z_\ell$ with covariance matrix $2h_\ell I_{2m}$. Define
$$F(t) : = S\bigg( \big( \rho^{\boxplus [n]} \star_{\frac{n'}{n}} X^{\boxplus [n']} \big) \star_{\frac{n''}{n}t}  Z^{\boxplus [n'']} \bigg) - \sum_{(\ve, \w) \in \mathcal{C}} \mu_{\ve, \w} S\bigg( \big( \rho^{\boxplus \ve} \star_{\frac{|\w|}{|\ve|}} X^{\boxplus \w}\big) \star_{\frac{|\w'_{(\ve, \w)}|}{|\ve|}t} Z^{\boxplus \w^\prime_{(\ve,\w)}}\bigg).$$
We note that the covariance matrix of $Z^{\boxplus \w'_{(\ve, \w)}}$ is equal to $2b_{(\ve, \w)}I_{2m}$ where
$$b_{(\ve, \w)}=\frac{1}{|\w'_{(\ve, \w)}|}\sum_{\ell\in \w'_{(\ve, \w)}} h_\ell = \frac{r\mu_{\ve, \w} n}{|\w'_{(\ve, \w)}|},$$
and similarly the covariance matrix of $Z^{\boxplus [n'']}$ is $2b_{[n'']}I_{2m}=2\frac{n}{n''}I_{2m}$.
Therefore, the quantum de Bruijn identity~\eqref{eq:deBruijn} yields
\begin{align}
\frac{\dd}{\dd t} F(t) & =  \frac{n''}{n} b_{[n'']}I_{\KMB}\bigg( \big( \rho^{\boxplus [n]} \star_{\frac{n'}{n}} X^{\boxplus [n']} \big) \star_{\frac{n''}{n}t}  Z^{\boxplus [n'']} \bigg)\nonumber \\
& \qquad  - \sum_{(\ve, \w) \in \mathcal{C}} \mu_{\ve, \w} \frac{|\w'_{(\ve, \w)}|}{|\ve|} b_{(\ve, \w)} I_{\KMB}\bigg( \big( \rho^{\boxplus \ve} \star_{\frac{|\w|}{|\ve|}} X^{\boxplus \w}\big) \star_{\frac{|\w'_{(\ve, \w)}|}{|\ve|}t} Z^{\boxplus \w^\prime_{(\ve,\w)}}\bigg)\nonumber \\
& = I_{\KMB}\bigg( \big( \rho^{\boxplus [n]} \star_{\frac{n'}{n}} X^{\boxplus [n']} \big) \star_{\frac{n''}{n}t}  Z^{\boxplus [n'']} \bigg)\nonumber \\
& \qquad  - \sum_{(\ve, \w) \in \mathcal{C}}  \frac{r \mu_{\ve, \w}^2 n}{|\ve|}  I_{\KMB}\bigg( \big( \rho^{\boxplus \ve} \star_{\frac{|\w|}{|\ve|}} X^{\boxplus \w}\big) \star_{\frac{|\w'_{(\ve, \w)}|}{|\ve|}t} Z^{\boxplus \w^\prime_{(\ve,\w)}}\bigg)\label{eq:Aux_Entropy_KMB}.
\end{align}
Define $\hat X_1, \dots, \hat X_{n'+n''}$by 
\begin{align*}
\hat X_\ell =\begin{cases} 
X_{\ell} ~~ & 1\leq \ell \leq n',\\
 \sqrt t Z_{\ell-n'} ~~  & n'+1\leq \ell \leq n'+n''.
\end{cases}
\end{align*}
Then, we have 
$$ \big( \rho^{\boxplus [n]} \star_{\frac{n'}{n}} X^{\boxplus [n']} \big) \star_{\frac{n''}{n}t}  Z^{\boxplus [n'']}  = \rho^{\boxplus [n]}\star_{\frac{n'+n''}{n}} \hat X^{[n'+n'']},$$
and 
$$\big( \rho^{\boxplus \ve} \star_{\frac{|\w|}{|\ve|}} X^{\boxplus \w}\big) \star_{\frac{|\w'_{(\ve, \w)}|}{|\ve|}t} Z^{\boxplus \w^\prime_{(\ve,\w)}} =  \rho^{\boxplus \ve} \star_{\frac{|\w|+|\w'_{(\ve, \w)}|}{|\ve|}} \hat X^{\boxplus \w\dot\cup \w'_{(\ve, \w)}},$$
where we understand $\w\dot\cup \w'_{(\ve, \w)}$ as $\w\cup \{\ell +n':\, \ell \in \w'_{(\ve, \w)}\}$. Hence, 
\begin{align*}
\frac{\dd}{\dd t} F(t) 
& = I_{\KMB}\bigg( \rho^{\boxplus [n]}\star_{\frac{n'+n''}{n}} \hat X^{[n'+n'']}\bigg) - \sum_{(\ve, \w) \in \mathcal{C}}  \frac{r \mu_{\ve, \w}^2 n}{|\ve|}  I_{\KMB}\bigg( \rho^{\boxplus \ve} \star_{\frac{|\w|+|\w'_{(\ve, \w)}|}{|\ve|}} \hat X^{\boxplus \w\dot\cup \w'_{(\ve, \w)}}\bigg).
\end{align*}
and Theorem~\ref{Th:Shannon_Stam_KMB} implies $\frac{\dd}{\dd t} F(t)\leq 0$ for all $t\geq 0$. As a result, 
\begin{align*}
F(0)&  \geq \lim_{t\to +\infty} F(t) \\
&= m\log\Big(\frac{n''}{n}b_{[n'']}\Big) - m \sum_{(\ve, \w)\in \mathcal C} \mu_{\ve, \w}\log\bigg( \frac{|\w'_{(\ve, \w)}|}{|\ve|} b_{(\ve, \w)}\bigg)\\
&= m \sum_{(\ve, \w)\in \mathcal C} \mu_{\ve, \w}\log\frac{|\ve|}{r\mu_{\ve, \w} n},
\end{align*}
where the computation of the limit in the second line is based on Lemma~\ref{lem:entropy-heat-semigroup}. As a conclusion, we have
\begin{align} \label{eq:Generalized_Shannon_Stam}
S\Big(  \rho^{\boxplus [n]} \star_{\frac{n'}{n}} X^{\boxplus [n']} \Big) \geq  \sum_{(\ve, \w) \in \mathcal{C}} \mu_{\ve, \w} S \Big( \rho^{\boxplus \ve} \star_{\frac{|\w|}{|\ve|}} X^{\boxplus \w}\Big) + m \sum_{(\ve, \w)\in \mathcal C} \mu_{\ve, \w}\log\frac{|\ve|}{r\mu_{\ve, \w} n}.
\end{align}

Now let 
$$\mu_{\ve, \w}=\frac{1}{\alpha}|\ve| \exp \Big( \frac 1m  S \Big( \rho^{\boxplus \ve} \star_{\frac{|\w|}{|\ve|}} X^{\boxplus \w}\Big) \Big),$$
where $\alpha=\sum_{(\ve, \w) \in \mathcal{C}} |\ve| \exp \Big( \frac 1m  S \Big( \rho^{\boxplus \ve} \star_{\frac{|\w|}{|\ve|}} X^{\boxplus \w}\Big) \Big)$ is the normalization factor. If $r\mu_{\ve, \w}\leq 1$ for all $(\ve , \w)\in \mathcal W$, then using   
$$\frac{1}{m} S \Big( \rho^{\boxplus \ve} \star_{\frac{|\w|}{|\ve|}} X^{\boxplus \w}\Big)  = \log \frac{\alpha \mu_{\ve, \w}}{|\ve|},$$
in~\eqref{eq:Generalized_Shannon_Stam} we find that
$$\frac{1}{m}S\Big(  \rho^{\boxplus [n]} \star_{\frac{n'}{n}} X^{\boxplus [n']} \Big) \geq\log  \frac{\alpha}{rn},$$
which is equivalent to~\eqref{eq:Generalized_quantum_Classical_EPI}. 

If there is $(\ve_0 , \w_0)\in \mathcal C$ for which $r\mu_{\ve_0, \w_0} > 1$ or equivalently $\alpha< r|\ve_0| \exp \Big( \frac 1m  S \Big( \rho^{\boxplus \ve_0} \star_{\frac{|\w_0|}{|\ve_0|}} X^{\boxplus \w_0}\Big) \Big)$, then we have
\begin{align*}
\frac{\alpha}{rn} & \leq \frac{|\ve_0|}{n}\exp \Big( \frac 1m  S \Big( \rho^{\boxplus \ve_0} \star_{\frac{|\w_0|}{|\ve_0|}} X^{\boxplus \w_0}\Big) \Big)\\
& \leq  \frac{|\ve_0|}{n }\exp \Big( \frac 1m  S \Big( \rho^{\boxplus [n]} \star_{\frac{n'}{n}} X^{\boxplus [n']}\Big) -\log \frac{|\ve_0|}{n} \Big) \\
&= \exp \Big( \frac 1m  S \Big( \rho^{\boxplus [n]} \star_{\frac{n'}{n}} X^{\boxplus [n']}\Big)\Big),
\end{align*}
where the second line follows from~\eqref{eq:Generalized_Shannon_Stam} for the collection $\mathcal C' = \{(\ve_0, \w_0)\}$ with the trivial distribution $\mu'_{\ve_0, \w_0}=1$ and the corresponding value $r'=1$.

\subsection{Proof of Theorem~\ref{Th:Shannon_Stam_KMB}}\label{Sec:ProofStamKMB}
Lemma~\ref{lem:Fisher-finite} and the integral representation~\eqref{eq:HiaRuskai} imply that 
\begin{align*}
I_\KMB(\rho) =  \sup_{\epsilon\in (0,1]}  \frac 12 \sum_{j=1}^{m} \int_{\epsilon}^1 \Big(\, \Big\| S_{\rho, j}^{1,t} \Big\|_{\rho,1, t}^2 + \Big\| S_{\rho, j}^{2,t} \Big\|_{\rho,2, t}^2 \Big) \dd \omega(t).
\end{align*}
Therefore, $I_\KMB(\rho)$, although may be infinite, is well-defined. Moreover, to prove the theorem,
it is sufficient to show that
$$
\Big\|S_{\rho^{\boxplus [n]} \star_{\frac{n'}{n}} X^{\boxplus [{n^\prime}]}, j}\Big\|_{\rho^{\boxplus [n]} \star_{\frac{n'}{n}} X^{\boxplus [{n^\prime}]}}^2 \leq rn \sum_{(\ve, \w) \in \mathcal{C}} \frac{1}{| \ve |}  \mu_{(\ve,\w)}^2  \Big\|S_{\rho^{\boxplus \ve} \star_{\frac{|\w|}{|\ve|}} X^{\boxplus \w}, j}\Big\|_{\rho^{\boxplus \ve} \star_{\frac{|\w|}{|\ve|}} X^{\boxplus \w}}^2,
$$
where the inner products and the score operators are defined with respect to any of the linear inner products considered in previous sections. Fix such an inner product and 
let $\Pi$ be the orthogonal projection on the closure of the linear space $\big\{ \widetilde R_{[n], [{n^\prime}]}:\,   \| R \|_{\rho^{\boxplus [n]} \star X^{\boxplus [{n^\prime}]}} < +\infty   \big\}$. Then, by Corollary~\ref{corollary:projection-score-square} for any subsets $ \ve \subseteq [n]$ and $\w \subseteq [{n^\prime}]$ we have
\begin{equation*}
 \sqrt{\frac{n}{|\ve|}}  \Pi\big(\widetilde S_{\rho^{\boxplus \ve} \star_{\frac{|\w|}{|\ve|}} X^{\boxplus \w}, j} \big) =   \widetilde S_{\rho^{\boxplus [n]} \star_{\frac{n'}{n}} X^{\boxplus [{n^\prime}]}, j}.
\end{equation*}
As a result, using the fact that $\mu$ is a probability distribution on $\mathcal{C}$, we have 
\begin{align}\label{eq:projection}
  \widetilde S_{\rho^{\boxplus [n]} \star_{\frac{n'}{n}} X^{\boxplus [{n^\prime}]}, j} =   \Pi \bigg( \sum_{(\ve, \w) \in \mathcal{C}}  \mu_{(\ve,\w)}  \sqrt{\frac{n}{|\ve|}} \widetilde S_{\rho^{\boxplus \ve} \star_{\frac{|\w|}{|\ve|}} X^{\boxplus \w}, j} \bigg).
\end{align}
Moreover, using the fact that  $T\mapsto \widetilde T_{[n], [n']}$ is an isometry yields
\begin{align*}
\Big\| S_{\rho^{\boxplus [n]} \star_{\frac{n'}{n}} X^{\boxplus [{n^\prime}]}, j}\Big\|_{\rho^{\boxplus [n]} \star_{\frac{n'}{n}} X^{\boxplus [{n^\prime}]}}^2
& = \Big\|   \widetilde S_{\rho^{\boxplus [n]} \star_{\frac{n'}{n}} X^{\boxplus [{n^\prime}]}, j} \Big\|_{\rho^{\otimes [n]}, X^{[n^\prime]}}^2 \\
& = \bigg\| \Pi \bigg( \sum_{(\ve, \w) \in \mathcal{C}}  \mu_{(\ve,\w)}  \sqrt{\frac{n}{|\ve|}} \widetilde S_{\rho^{\boxplus \ve} \star_{\frac{|\w|}{|\ve|}} X^{\boxplus \w}, j} \bigg) \bigg\|_{\rho^{\otimes [n]}, X^{[n^\prime]}}^2\\
&\leq \bigg\|  \sum_{(\ve, \w) \in \mathcal{C}}  \mu_{(\ve,\w)}  \sqrt{\frac{n}{|\ve|}} \,\widetilde S_{\rho^{\boxplus \ve} \star_{\frac{|\w|}{|\ve|}} X^{\boxplus \w}, j}   \bigg\|_{\rho^{\otimes [n]}, X^{[n^\prime]}}^2,
\end{align*}
Next, using~\eqref{VarianceDrop} we can write
\begin{align*}
&\Big\| S_{\rho^{\boxplus [n]} \star_{\frac{n'}{n}} X^{\boxplus [{n^\prime}]}, j}\Big\|_{\rho^{\boxplus [n]} \star_{\frac{n'}{n}} X^{\boxplus [{n^\prime}]}}^2 \\
\quad 
& \leq  \bigg\|  \sum_{(\ve, \w) \in \mathcal{C} , (\ve^\prime, \w^\prime) \subseteq (\ve , \w)}  \mu_{(\ve,\w)}  \sqrt{\frac{n}{|\ve|}}  \, \mathcal{P}_{\ve^\prime, \w^\prime} \big( \widetilde S_{\rho^{\boxplus \ve} \star_{\frac{|\w|}{|\ve|}} X^{\boxplus \w}, j} \big)   \bigg\|_{\rho^{\otimes [n]}, X^{[n^\prime]}}^2 \\
&= \bigg\| \sum_{(\ve^\prime , \w^\prime) \in \mathcal{C}} \mathcal{P}_{\ve^\prime, \w^\prime} \bigg( \sum_{ (\ve , \w) \supseteq (\ve^\prime, \w^\prime)}  \mu_{(\ve,\w)}  \sqrt{\frac{n}{|\ve|}} \,  \widetilde S_{\rho^{\boxplus \ve} \star_{\frac{|\w|}{|\ve|}} X^{\boxplus \w}, j} \bigg)   \bigg\|_{\rho^{\otimes [n]}, X^{[n^\prime]}}^2 \\
&= \sum_{(\ve^\prime , \w^\prime) \in \mathcal{C}} \bigg\| \sum_{ (\ve , \w) \supseteq (\ve^\prime, \w^\prime)}  \mu_{(\ve,\w)}  \sqrt{\frac{n}{|\ve|}} \,  \mathcal{P}_{\ve^\prime, \w^\prime} \bigg( \widetilde S_{\rho^{\boxplus \ve} \star_{\frac{|\w|}{|\ve|}} X^{\boxplus \w}, j} \bigg)   \bigg\|_{\rho^{\otimes [n]}, X^{[n^\prime]}}^2 \\
&\leq r\sum_{(\ve^\prime , \w^\prime)\in \mathcal{C}}  \sum_{ (\ve , \w) \supseteq (\ve^\prime, \w^\prime)} \mu_{(\ve,\w)}^2  \frac{n}{|\ve|}   \bigg\| \mathcal{P}_{\ve^\prime, \w^\prime} \bigg( \widetilde S_{\rho^{\boxplus \ve} \star_{\frac{|\w|}{|\ve|}} X^{\boxplus \w}, j} \bigg)   \bigg\|_{\rho^{\otimes [n]}, X^{[n^\prime]}}^2  \\
&= r\sum_{(\ve , \w)\in \mathcal{C}}  \mu_{(\ve,\w)}^2  \frac{n}{|\ve|}  \sum_{ (\ve^\prime, \w^\prime) \subseteq (\ve , \w)} \bigg\| \mathcal{P}_{\ve^\prime, \w^\prime} \bigg( \widetilde S_{\rho^{\boxplus \ve} \star_{\frac{|\w|}{|\ve|}} X^{\boxplus \w}, j} \bigg)   \bigg\|_{\rho^{\otimes [n]}, X^{[n^\prime]}}^2 \\
& = rn\sum_{(\ve , \w)\in \mathcal{C}}  \mu_{(\ve,\w)}^2  \frac{1}{|\ve|}  \Big\|  \widetilde S_{\rho^{\boxplus \ve} \star_{\frac{|\w|}{|\ve|}} X^{\boxplus \w}, j} \Big\|_{\rho^{\otimes [n]}, X^{[n^\prime]}}^2 \\
& = rn \sum_{(\ve, \w) \in \mathcal{C}}   \mu_{(\ve,\w)}^2 \frac{1}{| \ve |}  \Big\|S_{\rho^{\boxplus \ve} \star_{\frac{|\w|}{|\ve|}} X^{\boxplus \w}, j}\Big\|_{\rho^{\boxplus \ve} \star_{\frac{|\w|}{|\ve|}} X^{\boxplus \w}}^2,
\end{align*}
where by $(\ve^\prime , \w^\prime) \subseteq (\ve, \w)$, we mean $\ve^\prime \subseteq \ve$ and $\w^\prime \subseteq \w$. Here, the fifth line follows by the Cauchy--Schwarz inequality considering the fact that there are at most $r$ elements in each term corresponding to any $(\ve^\prime , \w^\prime) \subseteq (\ve, \w)$ if $\ve^\prime$ and $\w^\prime$ are not both empty sets. We should note that if $\ve^\prime$ and $\w^\prime$ are empty sets, then $$\mathcal{P}_{\ve^\prime, \w^\prime} \big( \widetilde S_{\rho^{\boxplus \ve} \star_{\frac{|\w|}{|\ve|}} X^{\boxplus \w}, j} \big) = 0.$$
Also, in the penultimate line we once again use the orthogonal decomposition~\eqref{VarianceDrop}, and the last line follows from the fact that  $T\mapsto \widetilde T_{\ve, \w}$ is an isometry.

\bigskip

\paragraph{Acknowledgements.} SB is supported by the Iran National Science Foundation (INSF) under project No.~4031370. HM is partially supported by NSF grant No.~2329662.

\appendix
\section{Trace class operators are dense}\label{app:trace-dense}

In this appendix we show that trace class operators are dense in the space of operators $X$ satisfying $\|X\|_{\rho}<+\infty$. Here, the inner product $\langle \cdot, \cdot \rangle_{\rho}$ is defined in terms of one of the functions $\psi_{k, t}$ for $k= 1, 2$ and $t\in [0,1]$. We note that these functions take the form $\psi(x, y) = cx+dy$ with $c, d\geq 0$, so the corresponding inner product is given by 
$$\langle X, Y\rangle_{\rho} = c \tr( \rho Y X^\dagger) + d\tr(\rho X^\dagger Y).$$

Let $\rho = \sum_{i=0}^{+\infty} \lambda_i \ketbra{e_i}{e_i}$ be the spectral decomposition of $\rho$, and for any $k\geq 0$ define $P_k = \sum_{i=0}^k \ketbra{e_i}{e_i}$. We note that $P_\ell XP_k$ is finite-rank and trace class. Thus, to prove our claim it suffices to verify that for any $\epsilon>0$, there are $k, \ell$ such that  $\|X - P_\ell XP_k\|_{\rho}<\epsilon$.  

Since $P_k$'s commute with $\rho$, we have  
$$\langle X, P_\ell XP_k\rangle_{\rho} =\langle  P_\ell XP_k, X\rangle_{\rho} = \|P_\ell XP_k\|_{\rho}^2 = c \tr(\rho P_\ell X P_k X^\dagger) + d\tr(\rho  X^\dagger P_\ell X P_k ),$$
which implies 
\begin{align*}
\|X - P_\ell XP_k\|_{\rho}^2 & = \|X\|_\rho^2 - \|P_\ell XP_k\|_{\rho}^2 .
\end{align*}
Therefore, to prove the claim it suffices to show that $\sup_{k, \ell}  \|P_\ell XP_k\|_{\rho}^2 =\|X\|_\rho^2$. To this end, we compute 
\begin{align*}
\|X\|_\rho^2 & = c \sum_{i=0}^{+\infty} \lambda_i \|X^\dagger \ket{e_i}\|^2 + d \sum_{i=0}^{+\infty} \lambda_i \|X \ket{e_i}\|^2\\
& = c \sup_{\ell}  \sum_{i=0}^{\ell} \lambda_i \|X^\dagger \ket{e_i}\|^2 + d\sup_{k }\sum_{i=0}^{k} \lambda_i \|X \ket{e_i}\|^2\\
& = c \sup_{\ell} \sup_k \sum_{i=0}^{\ell} \lambda_i \|P_k X^\dagger \ket{e_i}\|^2 + d\sup_{k } \sup_\ell \sum_{i=0}^{k} \lambda_i \|P_\ell X \ket{e_i}\|^2\\
& = c \sup_{k, \ell } \tr(\rho X P_kX^\dagger P_\ell) + d\sup_{k, \ell }  \tr(\rho   X^\dagger P_\ell X P_k)\\
& =  \sup_{k, \ell }\Big(c \tr(\rho P_\ell X P_kX^\dagger ) + d\tr(\rho   X^\dagger P_\ell X P_k)\Big)\\
& = \sup_{k, \ell}  \|P_\ell XP_k\|_{\rho}^2.
\end{align*}

\section{A combinatorial design and a system of linear equations}\label{app:design}

\begin{lemma}\label{lem:subsets-linear} 
Let $0< r \leq P $ be two positive integers, and $\mu$ be a probability distribution on the set $[P]$ such that $r \mu_j \leq 1$ for any $j \in [P]$. Then, there is some integer $n \geq 1$, subsets $\w_j \subseteq [n]$ for every $j \in [P]$, and non-negative reals $h_\ell \geq 0$ for every $\ell \in [n]$, such that each element $\ell \in [n]$ appears in at most $r$ subsets $\w_j$, and 
\begin{align}\label{eq:CombDesignCon}
\sum_{\ell \in [n]} h_\ell  = 1, \qquad \text{ and }\qquad \sum_{\ell \in \w_j} h_\ell = r \mu_j, \quad \forall j \in [P].
\end{align}
\end{lemma}
\begin{proof}
Let $A_r$ to be the convex set in $\R^P$ consisting of probability distributions on $[P]$ that satisfy $r\mu_j \leq 1$ for every $j \in [P]$. Also, let $B_r$ to be the set of all distributions $\mu$ in $A_r$ such that we can find desired parameters that satisfy~\eqref{eq:CombDesignCon}. We claim that $B_r$ is also a convex set, which contains extreme points of $A_r$.

To see the convexity of $B_r$, let $\mu^1$ and $\mu^2$ be two elements in $B_r$, and $\lambda \in [0,1]$ be an arbitrary number. For any $k\in \{1,2\}$, we indicate the corresponding parameters for $\mu^k$ by integers $n_k$, subsets $\w_j^k \subseteq [n_k]$ for any $j \in [P]$, and non-negative reals $h_\ell^k \geq 0$ for any $\ell \in [n_k]$, that satisfy
\[
\sum_{\ell \in [n_k]} h_\ell^k  = 1, \qquad \text{ and }\qquad \sum_{\ell \in \w_j^k} h_\ell^k = r \mu_j^k, \quad \forall j \in [P],
\]
and each element $\ell \in [n_k]$ appears in at most $r$ subsets $\w_j^k$. Now let $\mu = \lambda \mu^1 + (1-\lambda) \mu^2$ be an element in $A_r$. To show that $\mu$ also belongs to $B_r$, we need to construct parameters that satisfy~\eqref{eq:CombDesignCon}. 

To do so, let $n = n_1 + n_2$, and for any $j \in [P]$ define $\w_j \subseteq [n]$ as $\w_j \coloneqq \w_j^1 \cup \{ n_1 + s; \: s \in \w_j^2\}$. Also for any $\ell\in [n]$, we construct non-negative reals $h_\ell \geq 0$ as
\begin{align} 
	h_\ell = \begin{cases}
		 \lambda h_\ell^1, \quad 1 \leq \ell \leq n_1\\
		(1-\lambda) h_{\ell - n_1}^2, \quad n_1 +1 \leq \ell \leq n_1+n_2. 
	\end{cases}
\end{align}
We observe that, since each element $\ell_k \in [n_k]$ appears in at most $r$ subsets $\w_j^k \subseteq [n_k]$, each $\ell \in [n]$ also appears in at most $r$ subsets $\w_j \subseteq [n]$ by definition. Moreover, for every $j \in [P]$ we have 
\[
\sum_{\ell \in \w_j} h_\ell = \lambda \sum_{\ell \in \w_j^1} h_\ell^1 + (1-\lambda) \sum_{\ell \in \w_j^2} h_\ell^2 = \lambda r  \mu_j^1 + (1-\lambda) r  \mu_j^2 = r \mu_j.
\]
It is also straightforward to see that $\sum_{\ell \in [n]} h_\ell  = 1$. As a result, $\mu$ also belongs to $B_r$, and we can conclude that $B_r$ is a convex subset of $A_r$. 

For the extreme points of $A_r$, we observe that they are precisely those distributions $\mu$ on $[P]$ such that $\mu_j = \frac 1r$ for $r$ elements in $[P]$, and $\mu_j=0$ for all remaining elements in $[P]$. To construct the corresponding parameters for these extreme points, we may simply take $n=1$,  $\w_j = [1]$ whenever $\mu_j \neq 0$, and $\w_j = \emptyset$ otherwise, and let $h_1 = 1$. Since there are only $r$ index which have non-zero $\mu_j$, we can see that the only element of $[1]$ appears in exactly $r$ subsets $\w_{j}$, and the condition~\eqref{eq:CombDesignCon} is satisfied.

Thus, $B_r$ is a convex subset in $A_r$ that contains extreme points of $A_r$, and since $A_r$ itself is a convex set, we conclude that $A_r = B_r$.
\end{proof}

{\small
\bibliographystyle{abbrvurl} 
\bibliography{CLTBIB}
}

\end{document}